\documentclass[floatfix,superscriptaddress,aps,nofootinbib,prx,reprint,longbibliography]{revtex4-2}
\usepackage[utf8]{inputenc}
\usepackage{amsmath}
\usepackage{amsfonts,dsfont}
\usepackage{amssymb,amsthm}
\usepackage{comment}
\usepackage{url}
\newcommand{\ket}[1]{|#1\rangle}
\newcommand{\bra}[1]{\langle#1|}

\usepackage{complexity}
\usepackage[dvipsnames]{xcolor}
\usepackage{physics}
\usepackage{mathpazo}
\usepackage{tikz}
\usetikzlibrary{quantikz}

\newtheoremstyle{mystyle}%
{3pt}
{3pt}
{\upshape}
{}
{\bfseries}
{.}
{.5em}
{}

\newcommand{\oo}{00\ldots0}
\newcommand{\po}{p_{\oo}}
\newcommand{\sP}{{\#\sf{P}}}
\usepackage[]{hyperref}
\hypersetup{
  colorlinks = true,
  urlcolor = Blue,
  linkcolor = red,
  citecolor = cyan
}

\usepackage[all]{hypcap}
\usepackage[capitalise,compress]{cleveref}
\crefname{section}{Section}{Section} 
\crefname{subsection}{Section}{Section}
\theoremstyle{mystyle}
\crefname{thm}{Theorem}{Theorems}
\Crefname{thm}{Theorem}{Theorems}
\newtheorem{thm}{Theorem}
\newtheorem{lemma}{Lemma}

\newtheorem{corll}{Corollary}

\crefname{corll}{Corollary}{Corollaries}
\theoremstyle{remark}

\theoremstyle{definition}
\newtheorem{defn}{Definition}
\theoremstyle{mystyle}

\begin{document}

\title{Tight bounds on the convergence of noisy random circuits to the uniform distribution}
\author{Abhinav Deshpande}
\affiliation{Joint Center for Quantum Information and Computer Science and Joint Quantum Institute, NIST/University of Maryland, College Park, Maryland 20742, USA}
\affiliation{Institute for Quantum Information and Matter, Caltech, Pasadena, California 91125, USA}

\author{Pradeep Niroula}
\affiliation{Joint Center for Quantum Information and Computer Science and Joint Quantum Institute, NIST/University of Maryland, College Park, Maryland 20742, USA}

\author{Oles Shtanko}
\thanks{Currently at IBM Quantum, MIT-IBM Watson AI Lab, Cambridge MA, 02139 US}
\affiliation{Joint Center for Quantum Information and Computer Science and Joint Quantum Institute, NIST/University of Maryland, College Park, Maryland 20742, USA}

\author{Alexey V. Gorshkov}
\affiliation{Joint Center for Quantum Information and Computer Science and Joint Quantum Institute, NIST/University of Maryland, College Park, Maryland 20742, USA}

\author{Bill Fefferman}
\affiliation{Department of Computer Science, University of Chicago, Chicago, Illinois 60637, USA}

\author{Michael J. Gullans}
\affiliation{Joint Center for Quantum Information and Computer Science and Joint Quantum Institute, NIST/University of Maryland, College Park, Maryland 20742, USA}

\begin{abstract}
We study the properties of output distributions of noisy, random circuits.
We obtain upper and lower bounds on the expected distance of the output distribution from the ``useless'' uniform distribution.
These bounds are tight with respect to the dependence on circuit depth.
Our proof techniques also allow us to make statements about the presence or absence of anticoncentration for both noisy and noiseless circuits.
We uncover a number of interesting consequences for hardness proofs of sampling schemes that aim to show a quantum computational advantage over classical computation.
Specifically, we discuss recent barrier results for depth-agnostic and/or noise-agnostic proof techniques.
We show that in certain depth regimes, noise-agnostic proof techniques might still work in order to prove an often-conjectured claim in the literature on quantum computational advantage, contrary to what was thought prior to this work.
\end{abstract}

\maketitle

\section{Introduction}
Noise is an unavoidable part of any quantum computing experiment today.
The importance of considering the limitations of noisy quantum computers is most palpable in the coinage and the popularity of the term \emph{noisy, intermediate-scale quantum (NISQ)} computers \cite{Preskill2018}.
Because of these limitations, the study of quantum algorithms and of their robustness to noise is a problem at the forefront of quantum information science today.
With regard to the aim of outperforming classical computers at solving computational problems, two obstacles are noise and limited system size.
There is a tradeoff between the two, since it is generally challenging to realize a quantum computation with both large enough number of qubits and low enough error rates.
This is the reason why recent demonstrations of ``quantum computational advantage'' \cite{Arute2019,Zhong2020,Zhong2021,Wu2021} have been rightly hailed as exciting developments.
It is of prime importance in the field of quantum information today to study the tradeoff between system size and noise from the viewpoint of whether a given experiment is indeed efficiently simulable on a classical computer or not.

Separately, in recent years, the study of random quantum circuits has seen renewed vigor, because of their ability to model chaotic \cite{Brown2012,Brown2015} and complex \cite{Brandao2021,Haferkamp2021a} quantum dynamics and the ability to study certain universal properties of subclasses of random circuits using methods from statistical physics \cite{Potter2021}.
Indeed, random circuits and random ans\"atze sometimes inform the modeling of near-term variational quantum algorithms, an example being the barren-plateau problem \cite{McClean2018,Wang2021}.

In this work, we study various properties of noisy random circuits relating to their rate of convergence to the uniform distribution.
Specifically, we study circuits of depth $d$ on $n$ qubits with Haar-random two-qubit gates and local Pauli noise, with measurements in the computational basis at the output.
The precise rate of convergence of the resulting output distribution to the uniform distribution is a question of much significance in the complexity of random circuit sampling \cite{Gao2018,Bouland2021}, the theory of benchmarking noisy circuits \cite{Aharonov1996,Aharonov2000,Emerson2005,Boixo2018,Boixo2017,Liu2021}, and the investigation of near-term algorithms \cite{StilckFranca2021,Wang2021,Wang2021a}.
We prove upper and lower bounds on the expected total variation distance $\delta$ of the output distribution (when measuring in the computational basis) to the uniform distribution, which take the form $\delta \sim \exp[-\tilde{\Theta}(d)]$
(see note\footnote{In this paper, we use the symbols $O$, $\Omega$, $\Theta$, $o$, and $\omega$ to denote various relations between the asymptotic scaling of two nonnegative functions $f(n)$ and $g(n)$, studied in the limit $n \to \infty$.
We denote $f=O(g)$ if $\lim_{n\to \infty} f(n)/g(n) < \infty$, which is also equivalent to the relation $g = \Omega(f)$.
Also, if $\lim_{n\to \infty} f(n)/g(n) = 0$, we have $f=o(g) \Leftrightarrow g = \omega(f)$.
Lastly, we say $f=\Theta(g)$ if both $f=O(g)$ and $f=\Omega(g)$ hold.
Symbols with a tilde, such as $\tilde{\Theta}$, suppresses potential polylogarithmic factors of $n$.
We also use the asymptotic notation in conjunction with the negative sign in the following sense: $-\Omega(f)$ denotes the set of all functions $g$ such that $-g(n) = \Omega(f(n))$, or $\lim_{n\to \infty} \frac{f}{-g} < \infty$.
Therefore, $g(n)= -n^2$ satisfies $g(n) = -\Omega(n)$.}).
These bounds are tight with respect to the scaling with $d$, in the sense that the exponent scales linearly with $d$.
We also study a property known as anticoncentration in noisy and noiseless random circuits.
Anticoncentration is a measure of ``flatness'' of the output distribution, which is why our results on the closeness to the uniform distribution inform anticoncentration properties.
Anticoncentration has been cited as crucial for output distributions to be classically hard to sample from \cite{Aaronson2013a,Harrow2017,Harrow2018} in the literature on quantum advantage via sampling tasks.

We first briefly describe our main results and their consequences.
\begin{itemize}
\item We prove a lower bound on the expected total variation distance of the output distribution from the uniform distribution for local Pauli noise, denoted $\mathbb{E}_\mathcal{B}[\delta]$ (\cref{thm_lowerbound_vardist})\footnote{The expectation value $\mathbb{E}_\mathcal{B}$ here is over the choice of the random unitary gates.}.
This takes the form $\mathbb{E}_\mathcal{B}[\delta] \geq \exp[-O(d)]$.
\item We prove an upper bound on the above quantity for the case of a stochastic Pauli noise channel we call \emph{heralded dephasing} (\cref{thm_upperbound_vardist}).
The upper bound takes the form $\mathbb{E}_\mathcal{B}[\delta] \leq \poly(n)\exp[-\Omega(d)]$.
\item We also study anticoncentration properties of noisy and noiseless random circuits.
We show that at sublogarithmic depth, there is a severe lack of anticoncentration, strengthening the results of \citet{Dalzell2020} (\cref{thm_no_anticonc}).
\item We complement the above result by showing that noisy random circuits with local Pauli noise do anticoncentrate at higher depth, since they anticoncentrate at least as fast as noiseless random circuits (\cref{thm_noisy_anticonc}).
\item As a result of independent interest, we develop a mapping between noisy random circuits and a model in statistical mechanics in the context of proving our results in \cref{apx_upperbound}.
This model builds upon tools invented in the study of  random circuits \cite{Nahum2017,Nahum2018,Hunter-Jones2019,Dalzell2020,Barak2020}.
\end{itemize}

Our results have important consequences for the tightness of proof techniques in the field of quantum computational advantage, which we discuss below.
Finally, we comment on the relation of our work to recently obtained results by \citet{Dalzell2021}, where bounds on the rate of convergence to the uniform distribution for noisy random circuits were obtained.  That work considers a low-noise limit where the local probability of error in each circuit location (denoted as $\epsilon$) satisfies the scaling  $\epsilon(n) n \to 0$ as the number of qubits $n$ tends to infinity.
Under these conditions, they are able to recover the scaling $\mathbb{E}_\mathcal{B}[\delta] \sim \exp[-\tilde{\Theta}(nd)]$ observed in small-scale numerics \cite{Boixo2017,Boixo2018}.
In contrast, we consider a more physically natural scaling limit where the noise rate stays constant in the large-$n$ limit.
In this case, the analysis of Ref.~\cite{Dalzell2021} breaks down.
On the other hand, our bounds on the convergence rate to uniformity apply in the limit of  vanishing noise rates, but the upper bound becomes uninformative.
Thus, the two sets of results provide  complementary and, apparently, largely non-overlapping insights into the behavior of noisy quantum circuits on finite systems.

\section{Consequences}

\subsection{Barriers on proof techniques in the complexity theory of random circuit sampling}
We now elaborate more on the connection between the hardness of sampling and the hardness of computing output probabilities for random quantum circuits.
There has been an effort in the literature to prove, under a reasonable complexity assumption, that approximately sampling from the output distribution of random quantum circuits is classically hard.
In order to prove this statement, it suffices to prove that approximating an output probability $\po := \abs{\bra{\oo}U\ket{\oo}}^2$ of an $n$-qubit random quantum circuit $U$ is hard on average \cite{Aaronson2013a,Bremner2016}.
More specifically, proving that $\po$ is $\sP$-hard to compute to within imprecision (measured in terms of the additive error) $2^{-n}/\poly(n)$ for some polynomial $\poly(n)$ would give the desired claim that approximately sampling from the target distribution to within a small imprecision is classically hard.
The target imprecision of $2^{-n}$ arises from the fact that the Hilbert space dimension is $2^n$ and a typical output probability of a given bitstring is $\sim 2^{-n}$.

The state-of-the-art results \cite{Krovi2022,Bouland2021,Kondo2021} on the average-case hardness of computing output probabilities of quantum circuits come close to proving the desired result in a certain sense.
The ``closeness'' is measured in terms of the largest imprecision to which computing the output probability $\po$ of a random circuit is still hard on average.
The state-of-the-art results prove that computing $\po$ is hard to within a smaller imprecision of $ 2^{-\Theta(n d)}$, matching the required imprecision of $2^{-n}/\poly(n)$ when $d$ is a constant.
These results improve upon prior results \cite{Bouland2019,Movassagh2019} that proved hardness with imprecision $ 2^{-\poly(n)}$.
Such results are often viewed as evidence for the conjecture that $\po$ is hard to compute on average to a much larger imprecision of $2^{-n}/\poly(n)$.

\subsubsection{Shallow depth random circuits}
In Ref.~\cite{Napp2019}, the authors gave important no-go results for proving the desired result, i.e.\ the average-case hardness of computing $\po$ to an imprecision $ 2^{-n}/\poly(n)$ for a specific class of constant-depth random circuits.
Specifically, they showed that it is in fact classically easy to compute $\po$ to within this much imprecision, even when previous techniques implied that computing $\po$ to within a much smaller imprecision of $2^{-\poly(n)}$ is average-case hard.
Therefore, these results mean that one cannot, in general, view the hardness of computing $\po$ to a smaller imprecision as evidence for the hardness of computing $\po$ to a larger imprecision.
In other words, these results constitute a \emph{barrier} for any technique purporting to prove the desired average-case hardness result for general quantum circuits.
Any such technique must necessarily be sensitive to the depth of the circuit, otherwise it would work for constant-depth circuits of the sort studied in Ref.~\cite{Napp2019} and contradict their easiness results.
Barrier results such as this are useful because they rule out certain proof techniques and guide the search for a proof technique resistant to these barriers.
In this case, the barrier result informs us about depth-sensitivity of a proof technique.

\subsubsection{Noisy random circuits}
A second barrier was identified by \citet{Bouland2021} concerning the issue of noise.
The authors showed that existing hardness proof techniques were applicable to noisy random circuits as well, to yield hardness of computing a noisy output probability $\po$ to small imprecision.
Contrastingly, for the slightly larger imprecision of $ 2^{-n}/\poly(n)$, it is known to be easy to compute output probabilities, since the output distribution in the presence of noise is believed to converge rapidly to the uniform distribution \cite{Aharonov1996,Boixo2018,Boixo2017}.
An algorithm that always outputs ``$1/2^n$'' successfully computes $\po$ to within the required imprecision.
Therefore, the results of Ref.~\cite{Bouland2021} exhibit another barrier for hardness proof techniques purporting to work with higher imprecision---these techniques must distinguish between noiseless and noisy random circuits.

\begin{figure*}[t!]
\centering
\includegraphics[width=\linewidth]{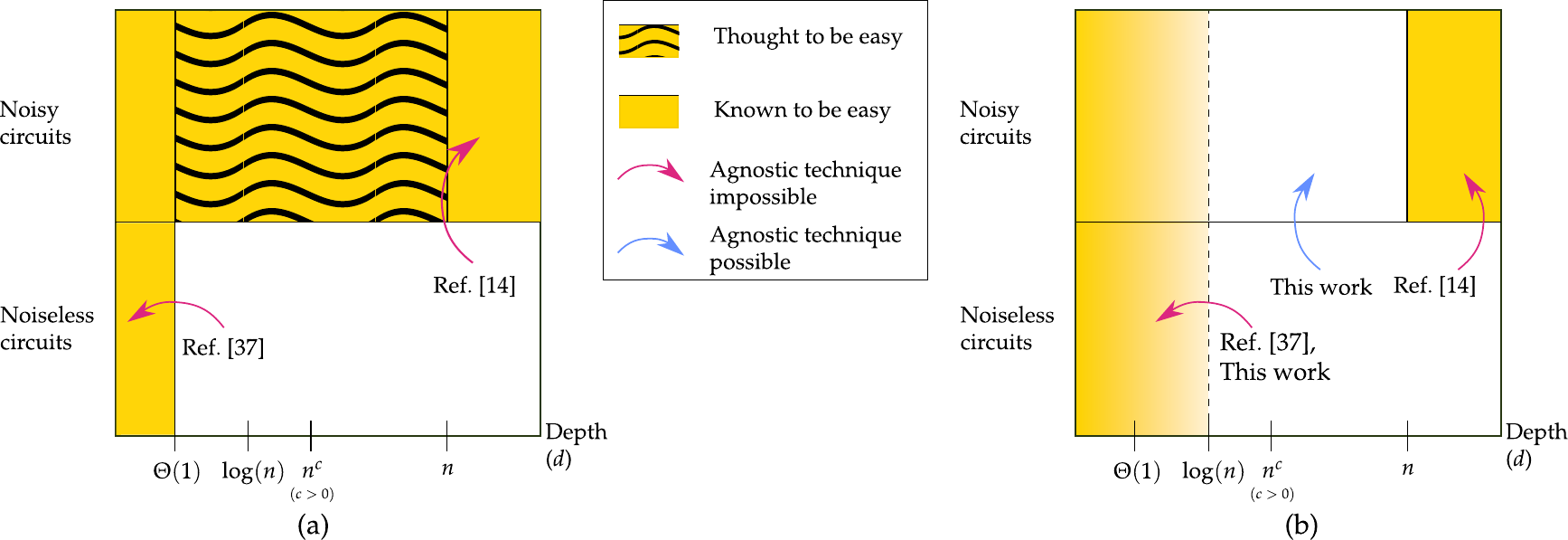} \label{fig_before_after}
\caption{
Status of hardness of approximating random output probabilities to within $2^{-an}$ imprecision for any constant $a>0$, (a) before and (b) after our work, for noisy and noiseless circuits at various depths.
The results in both (a) and (b) are shown for two dimensions\footnote{The results in \cref{fig_before_after}(b) remain the same in other dimensions.}.
The solid yellow regions correspond to the regions where the problem is known to be easy (for all $a$), while the dashed yellow regions are those where the problem was thought to be easy.
The red arrows imply that a technique to show hardness that is either noise-agnostic or depth-agnostic will fail because of the easy region, while the blue arrow indicates that such a technique is not ruled out.
The references next to the arrow refer to the works that discuss the presence or absence of a barrier.
In (a), the entire shaded region for noisy circuits follows from the assumption that noisy circuits converge to uniformity at a rate $2^{-\Theta(nd)}$.
Since we disprove this possibility in \cref{thm_lowerbound_vardist}, this region is smaller in (b) and allows for a noise-agnostic technique in the regime of $\Theta(\log n) \leq d \leq \Theta(n)$.
At large depths, noisy circuits continue to be easy due to their convergence to uniform, as strengthened by \cref{thm_upperbound_vardist}.
We also extend the easy region at shallow depths from constant to $o(\log n)$ by virtue of \cref{thm_no_anticonc}.
Finally, \cref{thm_noisy_anticonc} implies that the trivial algorithm of outputting 0 for noisy circuits stops working after depth $\Omega(\log n)$.
}
\end{figure*}

\subsubsection{Consequences of our results on tightness of proof techniques}
\citet{Bouland2021} also showed that existing noise-agnostic techniques for proving average-case hardness are almost tight.
The logic is that current techniques prove average-case hardness for imprecision $2^{-\Omega(nd \log nd)}$ or smaller.
On the other hand, assuming that noisy random circuits become $2^{-\Omega(n d)}$-close in total variation distance to the uniform distribution, as suggested by small-scale numerics \cite{Boixo2018,Boixo2017}, would mean that it is average-case easy to approximate output probabilities to within imprecision $2^{-O(nd)}$ or larger.
Thus, knowing the convergence properties of noisy random circuits sheds light on the tightness of noise-agnostic techniques.

We now discuss implications of our results on these aforementioned barrier results.
First, \cref{thm_no_anticonc} casts more light on the depth barrier result in Ref.~\cite{Napp2019}.
\citet{Napp2019} prove that for certain short depths $d\leq 3$ on certain architectures, approximating output probabilities to additive error $2^{-n}$ is easy on average.
This constitutes a barrier for improving robustness of proof techniques since the same output probabilities are formally average-case hard to approximate to within error $2^{-O(n \log n)}$.
Our \cref{thm_no_anticonc} states that for any sublogarithmic depth $d=o(\log n)$, most output probabilities are at most $2^{-n} \times 2^{-\Theta(n)} = o(2^{-n})$.
This indicates a severe lack of anticoncentration, a property that has been linked to the classical hardness of sampling from output distributions of random quantum circuits \cite{Aaronson2013a,Harrow2017,Harrow2018}.
Because of \cref{thm_no_anticonc}, a trivial algorithm that always guesses ``0'' as the output probability turns out to work to within imprecision $\leq 2^{-n}$ with high probability.
Therefore, we conclude that any technique to show the conjecture on average-case hardness to within imprecision $2^{-\Theta(n)}$ must not work at depth $d=o(\log n)$.
This extends the previously proven regime of $d \leq 3$ \cite{Napp2019}\footnote{\citet{Napp2019} also had non-rigorous arguments that appear to give convincing results for depths larger than 3 but smaller than some architecture-dependent constant.}.

Regarding the noise barrier, our result, namely \cref{thm_lowerbound_vardist}, shows that current noise-agnostic techniques may yet be improved.
This is because this theorem disproves the hypothesis that the total variation distance to uniformity follows $\mathbb{E}_\mathcal{B}[\delta] \leq \exp[-\Theta(nd)]$ in the asymptotic limit.
As mentioned earlier, this latter scaling was numerically observed at small system sizes  \cite{Boixo2017,Boixo2018} and hypothesized to hold asymptotically.
Since we show that the distance to uniform distribution behaves as $2^{-\tilde{\Theta}(d)}$, there is some scope for improving current noise-agnostic techniques.
Moreover, we also show in \cref{cor:typical_lower_bound} that, as long as the depth satisfies $d\leq c\log n$ for some constant $c$, the trivial algorithm ``output $1/2^n$'' fails with high probability.
This is achieved by a strengthening of \cref{thm_lowerbound_vardist}, where, in addition to a lower bound on the expectation value of the total variation distance $\mathbb{E}_\mathcal{B}[\delta]$, we show that this lower bound is \emph{typical}.

To summarize, noise-agnostic techniques can only work in the regime $d = \Omega(\log n)$, $d=O(n)$.
In fact, in this regime, the possibility of being able to prove hardness results for both noisy and noiseless circuits to imprecision $2^{-O(n)}$ is not ruled out (see \cref{fig_before_after}).

\subsubsection{Surmounting barriers}
Owing to this work, we have now identified a map of parameter regimes  showing where a quantum advantage may be obtained over classical computers and where depth- and noise-agnostic techniques can be used to prove the conjecture on average-case hardness of computing output probabilities to within imprecision at most $O(2^{-n})$.
The limitations are explored via two trivial algorithms for computing output probabilities---one always outputting $1/2^n$ and the other always outputting $0$.
We now speculate on how one might avoid, or surmount, these barrier results.

One possible lesson to glean from these results is that we need fundamentally new techniques to prove average-case hardness of computing output probabilities.
Indeed, all known average-case hardness results rely on polynomial interpolation, which was pioneered by \citet{Lipton1989}.
In the context of quantum advantage, all known proofs of average-case hardness \cite{Bouland2019,Movassagh2019,Bouland2021,Kondo2021,Krovi2022} construct a univariate polynomial in a specific way by perturbing each of the $\Theta(nd)$ gates of the random circuit and using the Feynman path integral.
This unavoidably leads to the imprecision depending on the combination $nd$ rather than on just $n$.
An alternative way of constructing polynomials that also explicitly fails for noisy circuits could potentially avoid this drawback and lead to a better dependence of the imprecision.

A second interpretation of our results relies on the following observation.
In both known barrier results, the trivial algorithm for computing output probabilities is just an algorithm that outputs a fixed constant independent of the input, i.e.\ either 0 or $1/2^n$.
Suppose we want to prove a hardness result saying that no $\PH$ algorithm that depends nontrivially on the input successfully approximates $\po$ on average.
Let us call these algorithms ``instance-dependent''\footnote{There are simple ways to make algorithms depend on the instance; we are interested in an algorithm that depends \emph{nontrivially} on the instance.}.
By our results, there are no barriers to proving such a hardness result on instance-dependent algorithms.
Therefore, it might be possible to obtain hardness results that only rule out instance-dependent algorithms that do not output a fixed value.
Thus, by finding a suitable enough technique, it may be possible to surmount both the noise and depth barriers.
Note that the algorithm of \citet{Napp2019} does not output a fixed constant. Nevertheless, our result showing that there is an alternative instance-independent algorithm in this regime raises the interesting open question of whether their algorithm fits into the framework of instance-independent algorithms we have identified here.

\subsection{Benchmarking noise using random circuits}
Sampling from random quantum circuits is a leading proposal for demonstrating a quantum computational advantage over classical computers \cite{Arute2019,Boixo2018,Aaronson2017,Bouland2019,Aaronson2020b}.
Part of the reason behind the strength of this proposal is the success of the linear cross-entropy measure as a predictor of fidelity~\cite{Boixo2018,Liu2021,Choi2021}, (although there can be exceptions \cite{Gao2021}).
This is believed to be a major advantage of schemes based on random circuit sampling over schemes based on other quantum sampling problems such as boson sampling and Gaussian boson sampling.
In fact, the resources needed for experimental implementation of benchmarking based on random circuit sampling can scale better than the implementation of randomized benchmarking \cite{Liu2021}.
The effectiveness of measures such as cross-entropy at reflecting the fidelity is crucially related to how the noisy distribution behaves and how close it is to the identity.
For example, a crucial fact used in Ref.~\cite{Liu2021} is that the average fidelity is upper and lower bounded by an exponential function of the noise strength ($e^{- c\epsilon n d} \leq \mathbb{E}[F] \leq Ke^{-c\epsilon n d}$ for known constants $c, K$)\footnote{Note that \citet{Liu2021} write this in terms of $\lambda = n\epsilon$, the total noise acting on a layer of gates.}; therefore, estimating the decay rate of the fidelity serves to estimate the noise strength $\epsilon$.
However, like Ref.~\cite{Dalzell2021}, this result is only applicable to the regime of asymptotically small noise strength $\epsilon \ll 1$.

In the more prevalent and natural regime of $\epsilon = \Theta(1)$, it would be worthwhile to obtain even tighter results on the scaling of the total variation distance between the experimental and uniform distributions with respect to the circuit depth.
Our work is the first step in obtaining asymptotically correct decay rates.
We prove tight bounds (with respect to scaling with depth) of the form $e^{-c_1 \epsilon d} \leq \mathbb{E}_\mathcal{B}[\delta] \leq Ke^{-c_2\epsilon d}$ for different $c_1, c_2$.
These results might be useful to show a similar decay behavior for other proxies for the fidelity, such as the linear cross entropy, which is an interesting question for future work.
Improving these bounds to obtain the same constants $c_1=c_2$ for sample-efficient estimators of the fidelity would lead to applications in benchmarking noise based on random circuits.

\subsection{Near-term algorithms}
Our results also have important consequences for near-term variational quantum algorithms that can be modeled by random circuits.
Variational algorithms with parametrized circuits suffer from the problem of ``barren plateaus'', which affects trainability of the circuit due to the vanishing of gradients in the cost function landscape \cite{McClean2018}.
Previous work on noise-induced barren plateaus \cite{Wang2021} has shown that the gradient of local cost functions vanishes as $\poly(n)\exp[-\Omega(d)]$ in the presence of noise.
One possible workaround to avoid the barren-plateau problem would be to use cost functions that are not (sums of) local observables and instead rely on postprocessing the entire data at the output distribution.
However, since we give an information-theoretic proof that the output distribution is close to the uniform one, our results (specifically \cref{thm_upperbound_vardist}) imply that even these strategies cannot ameliorate the problem.

More optimistically, at short depths $d=O(\log n)$, \cref{thm_lowerbound_vardist} implies that there is enough information content in the output distribution for circuit trainability.
We thus avoid a pessimistic conclusion that even constant-depth noisy circuits are untrainable, a conclusion that would have followed from the incorrect hypothesis that noisy circuits converge to uniformity at a rate $2^{-\Theta(nd)}$.

\subsection{Monitored random circuits and entanglement phase transitions} \label{sec_monitoredrqc}
As mentioned above, we study the upper bounds on the expected total variation distance to the uniform distribution for a noise channel called heralded dephasing.
Heralded dephasing corresponds exactly to the situation where, after each layer of the circuit, a random fraction $p$ of qubits is measured.
This is the model of monitored random circuits, or random circuits with intermediate measurements, that have been studied in the many-body-physics literature \cite{Aharonov2000,YungerHalpern2019,Chan2019,Li2018,Skinner2019,Gullans2020a,Potter2021}.

These models feature entanglement and purity phase transitions as a function of measurement strength (the probability $p$ of measuring a qubit at a given time), exhibiting a volume-law behavior of entanglement entropy of the final pure state for $p < p_c$, and an area-law behavior above the critical point ($p > p_c$).
These transitions have also been linked to computational complexity phase transitions \cite{Napp2019,Vijay2020,Bao2022}.
One can study the computational complexity of sampling from the output distribution of the random circuit provided knowledge of the measurement locations.
This problem can be studied both in the setting where the results of the intermediate measurements are known or unknown.
The case where the results of the intermediate measurements are unknown corresponds to heralded dephasing.
As a consequence of our work, we conclude that for the problem of approximate sampling from the output distribution of monitored random circuits with unknown intermediate measurement results, the problem becomes easy for classical computers for any depth $d = \omega(\log n)$ and any $p>0$.
Therefore, there is no complexity transition in this setting.
It remains open whether there is a complexity transition in the setting where the intermediate measurement results are also known.

\section{Prior work}
In this section, we will summarize prior work.
The foundational work by \citet{Aharonov1996} dealt with the convergence of states evolving under arbitrary quantum circuits with depolarizing noise to the maximally mixed state.
They showed that the entropy of the system reaches its maximal value exponentially fast with $d$.
They concluded that noisy circuits (without error correction) are essentially ``worthless'' after logarithmic depth $d=\Omega(\log n)$.
Translated to our setting, their proof techniques imply that the total variation distance $\delta$ to the uniform distribution satisfies $\delta \leq 2^{-\Omega(d)}$.
Depolarizing noise is also studied in more detail in Refs.~\cite{Muller-Hermes2016,Hirche2020}.
\citet{Ben-Or2013} generalized the work of \citet{Aharonov1996} to other forms of  noise and observed that the time after which noisy circuits are worthless depends on the class of channels.

More recently, \citet{Gao2018} studied the case of circuits with a more generalized form of Pauli noise and showed an upper bound on the distance to the uniform distribution of the form $\mathbb{E}[\delta] = \exp[-\Omega(d)]$.
This result, however, does not work for the case of dephasing noise.
In fact, the upper bound for dephasing noise is a constant independent of $n$ and $d$, which is not informative.
Moreover, the result explicitly assumes a property of random quantum circuits known as anticoncentration.

Note that, for dephasing noise, it is not possible to prove a general upper bound on $\delta$ that works for arbitrary circuits by using the techniques of Ref.~\cite{Gao2018}.
This is because a state in the computational basis is unaffected by dephasing, and hence there are instances (for example, circuits with diagonal gates in the computational basis) where the state is always in the computational basis throughout the evolution and the output distribution is unaffected by the noise channels.
Techniques like the Pauli twirl do not work either since we are interested in a quantity that is not a linear observable of $\rho$.

Prior work on anticoncentration has mostly been via second-moment bounds on the output probabilities \cite{Bremner2016,Harrow2018,Hangleiter2018,Barak2020,Dalzell2020}, which can be proved via the design property of random circuits \cite{Harrow2009a,Brandao2016,Harrow2018}.
An exception is the case of the one-clean qubit model (DQC1) \cite{Morimae2017b}.
A tool to analyze random circuits that has proven particularly fruitful is that of mapping to models in statistical mechanics.
This method has been successfully used in prior work; see, for example, \cite{Nahum2017,Nahum2018,Hunter-Jones2019,Dalzell2020,Barak2020}, among others.

\section{Definitions}
We define a depth-$d$ noisy  circuit on $n$ qubits as a sequence of quantum channels
\begin{align}
\mathcal{N}_d(\rho) = (\mathcal{E}_d \circ \mathcal{C}_d \cdots \mathcal{E}_1 \circ \mathcal{C}_1) (\rho),
\end{align}
where $\mathcal{C}_m(\rho) = U_m \rho U_m^\dag$ is a unitary operation on $n$ qubits and $\mathcal{E}_m(\rho)$ is a noise channel.
A noisy Haar random circuit is a noisy circuit for which each $U_m$ is decomposable into tensor products of two-qubit Haar random gates.

The results in our work are subject to different assumptions on the noise model.
We will always work with \emph{local} noise, where the channels $\mathcal{E}_m$ act on at most $k=\Theta(1)$ many qubits at a time.
We call a noise channel \emph{stochastic} if it can be written in Kraus form where at least one Kraus operator is a positive multiple of the identity.
\emph{Pauli} noise is described in general by $\mathcal{E}_m(\rho) = \sum_{E \in \mathcal{P}} p_m(E) E \rho E^\dag$, where $\mathcal{P}$ is the Pauli group on $n$ qubits and $p_m(E)$ is a probability distribution over the Pauli group elements.
The local noise rate in Pauli sector $\mu \in \{I, X, Y, Z\}$ at depth $m$ on site $i$  is the marginal distribution $q_{\mu mi}= \sum_{E \in \mathcal{P},E_i = \mu} p_m(E)$.
Most results in this work are applicable to Pauli noise channels, though we will sometimes discuss what other noise channels the proof techniques could be applicable to.
It is not too restrictive to consider Pauli noise channels since in practice one may use strategies like randomized compilation \cite{Wallman2016} to implement random circuit sampling or benchmarking based on random circuit sampling.
It is proved in Ref.~\cite{Wallman2016} that this technique tailors correlated and even coherent noise into stochastic Pauli noise.

We consider circuit architectures where, for simplicity, the gates are applied in parallel.
More formally, we define a \emph{parallel architecture} as one for which $n$ is an even number and every qubit is involved in a two-qubit gate at every unit of depth\footnote{We do not constrain the ``grouping'' of these qubits. In fact, this grouping can change from layer to layer.}.
Our results can be easily extended to more general gate layering strategies with a suitable redefinition of the depth.
For a given site $i$, we define $n_m(i)$ as the set of \emph{neighbors} of $i$, which are the sites involved in a two-qubit gate at circuit layer $m$ with $i$ including $i$ itself.
We extend $n_m$ to all subsets of sites $A \subset \{1,\ldots, n\}$ through the composition rule $n_m(A \cup B) = n_m(A) \cup n_m(B)$.
We define the \emph{forward} and \emph{backward} \emph{lightcones} $L_d(i)$ and $L_d^\dagger(i)$, respectively, of site $i$ at depth $d$ as the sets
\begin{align}
L_d(i) &= n_d\circ \cdots \circ n_1(i),\\
    L_d^\dagger(i) &=n_1 \circ \cdots \circ n_d(i).
\end{align}
In general, we have the bounds $|L_d(i)|, |L_d^\dagger(i)| \le 2^d$.

We denote the above ensemble of noisy parallel circuits as $\mathcal{B}$, and our results are generally stated in terms of expectation values over this ensemble, denoted $\mathbb{E}_\mathcal{B}$.
The ensemble $\mathcal{B}$ describes the distribution over unitaries only and not over the possible noise operations.
We also denote $\delta = \norm{\mathcal{D}- \mathcal{U}}_\mathrm{TVD}$, where $\mathcal{D}$ denotes the (noisy) output distribution when measuring in the computational basis, $\mathcal{U}$ is the uniform distribution over $n$-bit strings and ``TVD'' denotes total variation distance.
A central consequence of our results relates to anticoncentration properties of low-depth noisy random circuits.
Here, we provide a formal definition.
In the following, we denote by $\po$ the output probability of the string $\oo$.
For noiseless circuits, we have $\po =\abs{\bra{\oo} U \ket{\oo}}^2 $, while for noisy circuits it is given by $\po = \Tr[\ketbra{\oo} \mathcal{N}_d(\ketbra{\oo})]$.
\begin{defn} \label{def_antconc}
A family of random circuit ensembles is anticoncentrated if there exist constants $\alpha \in (0,1]$ and $c >0$ such that
\begin{align}
 \Pr_\mathcal{B}\left[\po
\geq \frac{\alpha}{2^n} \right] \geq c.
\end{align}  
\end{defn}
A stronger definition of anticoncentration is often given in terms of the \emph{collision probability}, which we define for noisy circuits as
\begin{equation}
    Z(U,\mathcal{E}) = \sum_{x\in \{0,1\}^n} p_x^2. \label{eq_collisionprob}
\end{equation}
We denote the collision probability for noiseless circuits as $Z(U, \mathbb{I})$.
The strong definition of anticoncentration requires that $2^n \mathbb{E}_\mathcal{B}[Z(U,\mathcal{E})] \leq c$ for some constant $c>0$ \cite{Barak2020,Dalzell2020}\footnote{Note, for a Haar random unitary $\mathbb{E}_\mathcal{U}  Z(U,\mathbb{I})= 2/(2^n+1)$.}.
It follows from a standard argument provided in the proof of \cref{thm_noisy_anticonc}, that, under this stronger definition of anticoncentration, the random circuit ensemble also satisfies \cref{def_antconc}; however,  the converse does not hold.

\section{Results}

\subsection{Lower bound on the distance to the uniform distribution}
\begin{thm} \label{thm_lowerbound_vardist}
For a Haar-random circuit on any parallel circuit architecture subject to local Pauli noise with a uniform upper bound on the local noise rate $q_{\mu mi} \leq q_\mu$ for all Pauli noise sectors $\mu$, $m$ and $i$, we have the lower bound
\begin{align}
\mathbb{E}_\mathcal{B}[ \delta] \geq \frac{(1-2b)^{2d}}{4\cdot 30^d},
\end{align}
where $b = \min[q_x + q_y, q_y+q_z, q_z+q_x] \le 1/2$.
\end{thm}
\begin{proof}

Let $p_x$ denote the probability of observing the bitstring $x$ at the output for a fixed circuit chosen from $\mathcal{B}$ and with a local noise channel applied after every gate.
For a region $A$, let $x_A$ be the value of $x$ restricted to the sites in $A$.
By definition, the total variation distance is the maximum difference in probabilities ascribed to any event $E$ by the two distributions, i.e.\ $\delta(\mathcal{D},\mathcal{U}) = \max_{E} [\abs{\Pr_\mathcal{D}[E] - \Pr_\mathcal{U}[E]}]$.
Considering the event that the first qubit is measured to be in the state $\ket{0}$, the 
total variation distance satisfies
\begin{align}
\delta \geq \abs{p_0 - \frac{1}{2}},
\end{align}
where $p_0$ and $p_1$ are the probabilities of observing a 0 and a 1 on the first qubit at the (noisy) output distribution, respectively.
Now since the above is a quantity in $[0,1]$, it satisfies
 $\abs{p_0 - \frac{1}{2}} \geq \abs{p_0 - \frac{1}{2}}^2$.
We now lower bound the expectation value of the latter quantity over parallel circuits, giving us $\mathbb{E}_\mathcal{B}[\delta] \geq \mathbb{E}_\mathcal{B}[p_0^2 - p_0 + 1/4]$.
Observe that for a given circuit,
\begin{align}
p_0 &= \Tr[\ketbra{0}_1 \mathcal{N}_d(\ketbra{\oo})]
\\ &= \sum_{x, x_1=0} \sum_{y, z, \ldots} \sum_{y', z', \ldots}\bra{x} V_d \ldots U_2 \ketbra{z} V_1 \ketbra{y} U_1  \ket{\oo} \nonumber \times
\\ & \bra{\oo} U_1^\dag \ketbra{y'} V_1^\dag \ketbra{z'} U_2^\dag \ldots V_d \ket{x}.
\end{align}
In the above, the gates $V_m$ are purifications of the noise channels $\mathcal{E}_m$ over ancillary qubits that are traced out.
The above equation is simply a Feynman-path integral representation of a noisy probability, from which it is evident that $p_0$ is a degree-two polynomial in the gate entries of the Haar-random gates.

This fact means that the expectation value $\mathbb{E}_\mathcal{B}[p_0^2 - p_0 + 1/4]$, which is over parallel circuits with local Haar-random gates, can be replaced by an expectation value over parallel circuits with local random Clifford gates $\mathbb{E}_C$ because of the 2-design property of the Clifford group.
Consider the quantity $\mathbb{E}_C[p_0^2 - p_0 + 1/4]$, which we reexpress as follows:
\begin{align}
\mathbb{E}_C[p_0^2 - p_0 + 1/4] = \frac{1}{4} - \mathbb{E}_C[p_0(1-p_0)].
\end{align}
We lower-bound this quantity by upper-bounding $\mathbb{E}_C[p_0(1-p_0)]$.
We define the probability density $f(p)$ through $f(p)\mathrm{d}p = \Pr_C(p_0 \in [p,p+\mathrm{d}p])$ and write
\begin{align}
\mathbb{E}_C[p_0(1-p_0)] &= \int_{0}^1 \mathrm{d}p\ f(p) p \cdot \left(1-p\right).  \label{eq_alternative}
\end{align}

Our overall strategy is as follows.
First, observe that, since the quantity $p(1-p)$ takes the maximum value $1/4$ in the interval $p \in [0,1]$, a crude upper bound for \cref{eq_alternative} is simply ${1}/{4} \times \int_0^1 \mathrm{d}p f(p) = {1}/{4}$.
This is a useless bound since it results in the conclusion $\mathbb{E}_\mathcal{B}[\delta] \geq 0$.
We refine this useless bound slightly by observing that at least some instances of the Clifford ensemble lead to a value of $p$ bounded away from $1/2$, implying that $p(1-p)$ is bounded away from $1/4$.
This will result in a better upper bound on $\mathbb{E}_C[p_0(1-p_0)]$, which will translate into a better lower bound on $\mathbb{E}_\mathcal{B}[\delta]$.

We split the integral in \cref{eq_alternative} into two parts: those with $p < \frac{1}{2} + \epsilon$ and those with $p \geq \frac{1}{2} + \epsilon$, for some $\epsilon \in (0, 1/2)$:
\begin{align}
\mathbb{E}_C[p_0 & (1-p_0)] =\int_0^{\frac{1}{2}+\epsilon} \mathrm{d}p f(p) p \cdot \left(1-p\right) \nonumber 
\\ & + \int_{\frac{1}{2}+\epsilon}^1 \mathrm{d}p f(p) p \cdot \left(1-p\right)
\nonumber
\\ & \leq \frac{1}{4} \int_0^{\frac{1}{2}+\epsilon} \mathrm{d}p f(p) + \left( \frac{1}{4} - \epsilon^2 \right) \int_{\frac{1}{2}+\epsilon}^1 \mathrm{d}p f(p).
\end{align}
In the above, we use the fact that, if $p \geq 1/2 + \epsilon$, then $p(1-p) \leq 1/4 - \epsilon^2$.
Continuing, we get
\begin{align}
\mathbb{E}_C[p_0(1-p_0)] &\leq \frac{1}{4}\left(1- \Pr_C\left(p_0 \geq \frac{1}{2} + \epsilon \right) \right) + \nonumber \\
\quad & \left( \frac{1}{4} - \epsilon^2 \right) \Pr_C\left(p_0 \geq \frac{1}{2} + \epsilon\right) 
\\ &= \frac{1}{4} - \epsilon^2 \Pr_C\left(p_0 \geq \frac{1}{2} + \epsilon \right).
\end{align}

It only remains to lower-bound $\Pr_C\left(p_0 \geq {1}/{2} + \epsilon \right)$.
For this we observe that we can take an extreme case over circuits satisfying a certain property.
As long as these circuits result in a final state with $p_0 \geq \frac{1}{2} + \epsilon$ and the likelihood of applying these (Clifford) circuits is large enough, we are done.
The extreme case we choose is simple: it consists of Clifford circuits where all of the $d$ two-qubit Clifford gates that touch the first qubit map the Pauli operator $Z_1$ to $Z_1$.
In particular, for Pauli noise channels this means that the first qubit is never entangled with any other qubit in the system and any mixture of the $\ketbra{0}$ and $\ketbra{1}$ states is unchanged by the unitary dynamics.
Effectively, the only evolution acting upon it is the noise channel after every layer.
Although these are  rare events, we will prove that they still lead to  a nontrivial lower bound  $\Pr_C\left(p_0 \geq {1}/{2} + \epsilon \right)$.

We will use \cref{lem_singlequbitnoise}, which analyzes this case and shows that for a single-qubit evolution under only Pauli noise, $\abs{p_0 - 1/2} = \left(1-2(q_x+q_y)\right)^d/2$.
This means that we can take $\epsilon = \left(1-2(q_x+q_y)\right)^d/2$.
Furthermore, the likelihood of observing this extreme event is lower-bounded away from 0. 
In each layer, the probability of applying a Clifford circuit with the property above is at least $1/30$ (this follows, for example, from a brute-force numerical evaluation for each of the 11520 two-qubit Clifford gates).
As a result, we obtain $\Pr_C(p_0 \geq 1/2 + \epsilon) \geq 1/30^d$ for $\epsilon = \left(1-2(q_x+q_y)\right)^d/2$.

Wrapping everything up, this results in
\begin{align}
\mathbb{E}_C[p_0(1-p_0)]  \leq \frac{1}{4}& - \frac{\left(1-2(q_x+q_y)\right)^{2d}/4}{30^d},
\\ \implies  \mathbb{E}_C[p_0^2 - p_0 + 1/4]  &\geq \frac{\left(1-2(q_x+q_y)\right)^{2d}}{4 \cdot 30^d} \\
\implies 
\mathbb{E}_\mathcal{B}[p_0^2 - p_0 + 1/4] &\geq \frac{\left(1-2(q_x+q_y)\right)^{2d}}{4 \cdot 30^d}, \label{eq:p0}
\\ \implies \mathbb{E}_\mathcal{B}[\delta] & \geq \frac{\left(1-2(q_x+q_y)\right)^{2d}}{4 \cdot 30^d}.
\end{align}
A comment is in order.
First, note that the dependence on $q_x + q_y$ can also be written as $q- q_z$, where $q=q_x+q_y+q_z$.
For the extreme event we have considered, it is understandable that dephasing noise (where $q=q_z$) does not affect the quantity $p_0$.
For these cases, we can strengthen the lower bound by considering events where the first gates act as Hadamards or Hadamards plus a phase gate $S$, the intermediate gates map $X_1$ to $X_1$ or $Y_1$ to $Y_1$, respectively, and the last gate inverts the first.
By symmetry, the previous analysis holds for these events as well.
We can combine everything to give the slightly better bound
\begin{align}
\mathbb{E}_\mathcal{B}[\delta] \geq \frac{\left(1-2b\right)^{2d}}{4 \cdot 30^d}, \label{eq:lowerbound}
\end{align}
where $b = \min[q_x + q_y, q_y + q_z, q_z + q_x] = q - \max[q_x, q_y, q_z]$.
\end{proof}
We also note that, for the case of perfect depolarizing noise on every qubit, we have $q_x = q_y = q_z = 1/4$.
This gives the trivial bound $\mathbb{E}_\mathcal{B}[\delta] \geq 0$, as it should, because perfect depolarizing noise immediately gives the identity operator on every qubit and the distance to the uniform distribution is exactly 0. Similarly, for a Haar random unitary, we have an estimate from the Porter-Thomas distribution $\Pr(p_{00\ldots 0}=p) = (2^n-1)(1-p)^{2^n-2}$ \cite{Boixo2018}, leading to
\begin{align}
\mathbb{E}_{U}[\delta] &= 2^{n-1}\int_0^1 dp \abs{p-\frac{1}{2^n}} \Pr(p_{00\ldots 0}=p) \ge e^{-1} .
\end{align}
Thus, in the case of noiseless Haar random circuits on architectures where the output probability approaches the Porter-Thomas distribution, our lower bound is expected to be a significant underestimate of the true value at large depths.
\citet{Barak2020} also obtain a similar lower bound for single-qubit marginal probabilities in the noiseless case.

\begin{lemma}\label{lem_singlequbitnoise}
Consider a single qubit starting in the state $\rho = \ketbra{0}$.
After $d$ applications of the channel $\mathcal{E}(\rho) = (1-q) \rho + q_x X\rho X + q_y Y \rho Y + q_z Z \rho Z$ (where $q = q_x + q_y + q_z$), the resulting state $\mathcal{E}^d(\rho)$ obeys
\begin{align}
p_0 := \Tr\left[ \frac{\mathds{1}+Z}{2}\mathcal{E}^d(\rho) \right] = \frac{1}{2} + \frac{1}{2}\left(1-2(q_x+ q_y)\right)^d.
\end{align}
\end{lemma}

\begin{proof}
It suffices to show that $\mathcal{E}^d(\rho) = p_0(d) \ketbra{0} + (1-p_0(d)) \ketbra{1}$, where $p_0(d) = \frac{1}{2} + \frac{1}{2}\left(1-2 (q_x + q_y)\right)^d$. We will prove this by induction. The $d = 0$ base case follows from the given initial state. The inductive step is
\begin{align}
\mathcal{E}^{d+1}&(\rho) = \mathcal{E}(\mathcal{E}^{d}(\rho)) \nonumber  \\
& = (1-q_x -q_y) [p_0(d) \ketbra{0} + (1-p_0(d)) \ketbra{1}] \nonumber \\ & \quad  + (q_x + q_y) [p_0(d) \ketbra{1} + (1-p_0(d)) \ketbra{0}] \nonumber \\
&  = p_0(d+1) \ketbra{0} + (1-p_0(d+1)) \ketbra{1},
\end{align}
which completes the proof.
\end{proof}

We also remark here that the lower bound (with potentially different constants) is applicable to all noise channels for which an analogue of \cref{lem_singlequbitnoise} holds.
The only condition we need is that after $d$ applications of the single-qubit noise channel to the initial state $\ketbra{0}$, the resulting state has $\abs{p_0(d) - 1/2} \geq \exp[-ad]$ for some $a>0$.
If this is satisfied, then one can take as an extreme case in the proof of \cref{thm_lowerbound_vardist} the Clifford circuits where all of the $d$ two-qubit Cliffords encountered by the first qubit act as identity on it.
This would change the denominator $30^d$ in \cref{eq:lowerbound} to $11520^d$, although similar symmetry arguments could improve this constant.
The same considerations apply to \cref{thm_no_anticonc}, whose proof makes use of the analysis in \cref{thm_lowerbound_vardist}.

As we mentioned in the introduction, \cref{thm_lowerbound_vardist} has wide-ranging consequences for the properties of noisy random circuits.
To strengthen the result, we now show that a similar lower bound applies to \emph{typical} noisy random circuits (i.e., individual elements of the ensemble that occur with high-probability) at sufficiently low depths.

\begin{corll}
\label{cor:typical_lower_bound}
For a noisy Haar-random circuit on any parallel circuit architecture  with a uniform upper bound on
the local noise rate $q_{\mu mi} \le q_\mu$ for all Pauli noise sectors $\mu$, $m$ and $i$,   we have the bound
\begin{equation}
    \Pr_{\mathcal{B}}(\delta < e^{-2 a d}) \le 8 e^{-a d}+16\cdot e^{(2a+\log 4)d}/n,
\end{equation}
where $a=-2 \log(1-2b) + \log(30)$ and $b = \min[q_x+q_y,q_y+q_z,q_z+q_x].$
\end{corll}
\begin{proof}
Following the arguments from \cref{thm_lowerbound_vardist}, we have the bound
\begin{equation}
    \delta \ge \Delta_i^2 := \frac{1}{4}\bra{0}\mathcal{N}^\dagger_d(Z_i) \ket{0}^2
\end{equation}
for every site $i$, where $\mathcal{N}_d^\dagger$ is the adjoint map of the depth-$d$ noisy circuit $\mathcal{N}_d$.
We can, thus, obtain a lower bound by the site-averaged quantity
\begin{equation}
\delta \ge \Delta^2 := \frac{1}{ n} \sum_i \Delta_i^2.    
\end{equation}
In \cref{thm_lowerbound_vardist}, we showed $\mathbb{E}_\mathcal{B}[\Delta_i^2]\geq e^{-ad}/4$, which implies $\mathbb{E}_\mathcal{B}[\Delta^2]\geq e^{-ad}/4$.
When the noise is given by a Pauli noise channel, as we consider, then the operator $\mathcal{N}_d^\dagger(Z_i)$ has support on at most $L_d^\dagger(i)$ sites.
We can use this fact to bound the second moment of the site-averaged quantity
\begin{align}
 \mathbb{E}_{\mathcal{B}}&[\Delta^4]  =   \frac{1}{n^2}\sum_{i,j}\mathbb{E}_{\mathcal{B}} [\Delta_i^2 
    \Delta_j^2], \\ 
    & = \frac{1}{n^2} \sum_{i;j \in L_{d}\circ L_d^\dagger(i)} \mathbb{E}_{\mathcal{B}} [\Delta_i^2 \Delta_j^2] \nonumber  \\ & \quad + \frac{1}{n^2} \sum_{i;j \notin L_{d}\circ L_d^\dagger (i)} \mathbb{E}_{\mathcal{B}}[\Delta_i^2] \mathbb{E}_{\mathcal{B}} [\Delta_j^2] \\
   &=   \frac{1}{n^2} \sum_{i;j \in L_{d}\circ L_d^\dagger(i)} (\mathbb{E}_{\mathcal{B}} [\Delta_i^2 \Delta_j^2] - \mathbb{E}_{\mathcal{B}}[\Delta_i^2] \mathbb{E}_{\mathcal{B}} [\Delta_j^2]) \nonumber \\ & \quad + \mathbb{E}_{\mathcal{B}}[\Delta^2]^2  \\
    & \le  \mathbb{E}_{\mathcal{B}}[\Delta^2]^2  + \frac{1}{n} \max_{i;j \in L_{d}\circ L_d^\dagger(i)} |L_{d}\circ L_d^\dagger(i)|  \sigma_i \sigma_j, \label{eqn:boundvar}
\end{align}
where we defined $\sigma_i^2 = \mathbb{E}_{\mathcal{B}}[\Delta_i^4] - \mathbb{E}_{\mathcal{B}}[\Delta_i^2]^2 \le 1$.
To see this, we use the Cauchy-Schwartz inequality:
\begin{align}
 \mathbb{E}_\mathcal{B}\left[(\Delta_i^2-\mathbb{E}_\mathcal{B}[\Delta_i^2])( \Delta_j^2-\mathbb{E}_\mathcal{B}[\Delta_j^2]) \right] \leq  \sigma_i \sigma_j, \label{eq_crosscorrelation}
\end{align}
whenever $j \in L_d \circ L_d^\dag(i)$.
For $j$ outside this set, the cross-correlation is zero.
Applying the lower bound from \cref{thm_lowerbound_vardist}, we now have the sequence of inequalities
\begin{align}
&\Pr_{\mathcal{B}}[\delta \le e^{-2 a d}] \leq \Pr_{\mathcal{B}}\left[\delta \le 4e^{-a d} \mathbb{E}_{\mathcal{B}}[\Delta^2]\right],  \\
& \leq \Pr_{\mathcal{B}}\left[\Delta^2 \le 4e^{-ad} \mathbb{E}_{\mathcal{B}}[\Delta^2] \right] , \\
&= 1- \Pr_{\mathcal{B}}\left[\Delta^2 > 4e^{-ad} \mathbb{E}_{\mathcal{B}}[\Delta^2]\right] ,\\
& \leq 1 - (1-4e^{-ad})^2 \frac{\mathbb{E}_{\mathcal{B}}[\Delta^2]^2}{\mathbb{E}_{\mathcal{B}}[\Delta^4]}, \label{eqn:paley} \\ \label{eqn:step2}
& \le 1 - \frac{ (1-4e^{-ad})^2\mathbb{E}_{\mathcal{B}}[\Delta^2]^2}{\mathbb{E}_{\mathcal{B}}[\Delta^2]^2 + \frac{1}{n} \max_{i;j \in L_{d}\circ L_d^\dagger(i)} [|L_{d}\circ L_d^\dagger(i)| \sigma_i \sigma_j]} ,\nonumber \\ 
&\le \frac{8 e^{-ad} \mathbb{E}_{\mathcal{B}}[\Delta^2]^2 + \frac{1}{n} \max_{i;j \in L_{d}\circ L_d^\dagger(i)} [|L_{d}\circ L_d^\dagger(i)| \sigma_i \sigma_j]}{\mathbb{E}_{\mathcal{B}}[\Delta^2]^2 + \frac{1}{n} \max_{i;j \in L_{d}\circ L_d^\dagger(i)} [|L_{d}\circ L_d^\dagger(i)| \sigma_i \sigma_j]},\\ \label{eqn:step}
& \leq \frac{8 e^{-ad} \mathbb{E}_{\mathcal{B}}[\Delta^2]^2 + \frac{1}{n} \max_{i;j \in L_{d}\circ L_d^\dagger(i)} [|L_{d}\circ L_d^\dagger(i)| \sigma_i \sigma_j]}{\mathbb{E}_{\mathcal{B}}[\Delta^2]^2} \\
& \leq 8e^{-ad} + 16\cdot 4^d e^{2 ad}/n,
\end{align}
where we applied the Paley-Zygmund inequality in \cref{eqn:paley}, used \cref{eqn:boundvar} in \cref{eqn:step2}, and used the fact that $|L_{d}\circ L_d^\dagger(i)| \le 4^d$ for every site $i$ to get \cref{eqn:step}.   
\end{proof}

\cref{cor:typical_lower_bound} shows that, for any depth that grows slower than 
\begin{equation}
    d < \left(\frac{1}{2a +\log 4}\right) \log n,
\end{equation}
for most circuits the total variation distance $\delta$ is lower bounded by $e^{-cd}$ for some fixed constant $c$.
This strengthens \cref{thm_lowerbound_vardist} by showing that the lower bound on $\mathbb{E}_\mathcal{B}[\delta]$ in \cref{thm_lowerbound_vardist} is actually characteristic of typical circuits at these depths.
The typicality result rules out the possibility that $\mathbb{E}_\mathcal{B}[\delta]$ is dominated by rare circuits with unusually large deviation $\delta$.
It is an interesting subject for future work to study similar typicality results for higher depths that scale polynomially with $n$.

\subsection{Upper bound on the distance to uniform}
In this section, we will study a partially heralded noise model where first a random set of sites are selected after each layer of the circuit independently with probability $p$.
At each site $i$ in this random subset, a local dephasing channel $\mathcal{E}_i$ is applied with dephasing parameter $q$, where
\begin{align}
\mathcal{E}_i(\rho) = (1-q) \rho + q Z_i \rho Z_i.
\end{align}
In the limit $p\to 1$, this becomes a standard local dephasing model with parameter $q$, while $q \to 1/2$ is equivalent to a model where a random set of sites are measured at rate $p$ in the $Z$-basis, but without keeping track of the measurement outcomes.
For $p<1$, we absorb the random locations of the dephasing events into the ensemble $\mathcal{B}$. 

The ``heralding'' refers to the fact that the set of sites where the measurements occurred is known, but not the measurement outcomes.
Note that this is different from a dephasing model where each site is uniformly dephased with dephasing parameter $pq$.
In particular, the noise locations act as an additional source of randomness in the model.

We focus on this noise model for two reasons.
First, we would like to verify the intuition that noise acting during a random unitary circuit renders the output distribution ``worthless'' (close to the uniform distribution) after logarithmic depth \cite{Aharonov1996}, even though there are atypical circuits that can avoid the effects of the noise in a heralded dephasing model.
The second motivation arises from the observation that when the measurement outcomes are also known, such models exhibit an entanglement transition in the conditional evolution of the quantum states as mentioned in \cref{sec_monitoredrqc} \cite{Potter2021}.
Our analysis of the heralded dephasing model proves that discarding the measurement outcomes, but maintaining knowledge of the noise locations, is enough to remove any signature of such entanglement transitions.
  
For a noisy Haar random circuit with this noise model, we prove an upper bound on the circuit-averaged total variation distance $\delta$ that is independent of the circuit architecture:

\begin{thm} \label{thm_upperbound_vardist}
For a noisy Haar random circuit on any parallel circuit architecture with heralded dephasing noise at rate $p$ with the dephasing parameter $q$, we have the upper bound
\begin{align}\mathbb{E}_{\mathcal{B}} [\delta] < \frac{3^{2/3}}{2}n^{1/3} e^{- \gamma p d/3},
\end{align}
where $\gamma = 8 q (1-q)/3 $. 
\end{thm}
\begin{proof}
We start from Pinsker's inequality, which states that the total variation distance between distributions $P$ and $Q$ is related to the corresponding KL-divergence (the classical relative entropy) by the relation $\delta(P, Q) \leq \sqrt{D_\mathrm{KL}(P||Q)/2}$.
The KL-divergence with respect to the uniform distribution $\mathcal{U}$ is given by $n\log 2 - H(P)$, where $H(P)$ is the Shannon entropy.
Let $\mathcal{D}$ denote the distribution of measurements for a circuit.
This gives us the following chain of inequality for variation distance to the uniform distribution:
\begin{align}
2\delta(\mathcal{D}, \mathcal{U})^2 &\leq D_{KL}\left(\mathcal{D}||\mathcal{U}\right), \\ &= n\log 2 - H(\mathcal{D}) \leq n\log 2 - H_2(\mathcal{D}),
\end{align}
\label{eq:upper-bond-tvd-squared}
where the last inequality follows from the fact that second R\'enyi entropy $H_2(\mathcal{D}) = -\log\left(\sum_{x}p_x^2\right)$ is less than or equal to the von-Neumann entropy $H(\mathcal{D}) = H_{\alpha \to 1}(\mathcal{D}) = -\sum_x p_x \log p_x$.
We can now use Markov's inequality to bound the average TVD.
For any $\epsilon \in [0,1]$, letting $\Pr_\mathcal{B}(\delta=\sigma)$ denote the probability \emph{density} of the continuous variable $\delta$, we have
\begin{align}
    \mathbb{E}_{\mathcal{B}}(\delta) &= {\int_{0}^{\epsilon}\mathrm{d}\sigma \sigma \Pr_{\mathcal{B}}(\delta=\sigma) + \int_{\epsilon}^{1}\mathrm{d}\sigma \sigma \Pr_{\mathcal{B}}(\delta=\sigma)}, \\
    &\leq \epsilon + {\Pr_{\mathcal{B}}(\delta \ge \epsilon)} \leq \epsilon + \frac{\mathbb{E}_{\mathcal{B}}(\delta^2)}{\epsilon^2}.
    \label{eq:tvd-expectation}
\end{align}
If $\mathbb{E}_{\mathcal{B}}(\delta^2)$ decays exponentially or faster with depth, i.e., $e^{-\gamma d}$ for some $\gamma > 0$, we can take $\epsilon = e^{-\gamma d/3}$ to ensure that  $\mathbb{E}_{\mathcal{B}}(\delta)$ decays exponentially as  $e^{-\gamma d/3}$.
To show that the second moment of TVD must indeed decay exponentially with $d$, we calculate the expectation of \cref{eq:upper-bond-tvd-squared}:
\begin{align}
   \mathbb{E}_{\mathcal{B}}[2\delta(\mathcal{D}, \mathcal{U})^2] &\leq n\log 2 - \mathbb{E}_{\mathcal{B}}(H_2(\mathcal{D})),  \nonumber \\ &= n\log2 + \mathbb{E}_{\mathcal{B}}\left[ \log\left(\sum_{x} p_x^2\right)\right], \nonumber \\ & \leq n\log2 + \log {\mathbb{E}_{\mathcal{B}}\left[\sum_{x} p_x^2\right]}.
\end{align}
In the last inequality, we have used Jensen's inequality for concave functions $\mathbb{E}_{\mathcal{B}}[f(X)] \leq f(\mathbb{E}_{\mathcal{B}}[X])$.
The term inside the expectation function, $\sum_x p_x^2$, is the collision probability.
From \cref{lem_collisionprob_upperbound}, we have that the expectation of the collision probability is upper-bounded by $2^{-n} \exp[\frac{n}{3}e^{-\gamma pd}]$, where $\gamma = 8q(1-q)/3$.
With this, we have
\begin{align} \nonumber
\mathbb{E}_{\mathcal{B}}[2\delta(\mathcal{D}, \mathcal{U})^2]  &\leq n\log 2 + \log \left[ 2^{-n}\exp\left[\frac{n}{3} e^{-\gamma pd} \right] \right] ,\\
& = \frac{n}{3} e^{-\gamma p d}.
\end{align}
Thus, we have that the second moment of the TVD decays exponentially in circuit depth.
The right hand side of \cref{eq:tvd-expectation} is minimized at $\epsilon = \left({n}/{3}\right)^{1/3} e^{-\gamma pd/3}$.
This yields the desired bound
\begin{align}
\mathbb{E}_{\mathcal{B}}(\delta) \leq \frac{3^{2/3}}{2}n^{1/3}e^{-\gamma pd/3}.
\end{align}
\end{proof}
To complete the proof, it remains to prove the following lemma:

\begin{lemma} \label{lem_collisionprob_upperbound}
For a noisy Haar random circuit on any parallel circuit architecture with heralded dephasing noise at rate $p$ with the dephasing parameter $q$, we have the upper bound on the collision probability $Z$
\begin{equation}
\mathbb{E}_{\mathcal{B}}[Z] = \mathbb{E}_{\mathcal{B}}\left[ \sum_x p_x^2 \right]\leq 
2^{-n} \exp[\frac{n}{3}e^{-\gamma pd}],
\end{equation}
where $\gamma=8 q (1-q)/3.$
\end{lemma}
To prove this bound, we make use of the statistical mechanics mapping method developed by \citet{Dalzell2020}.
The proof of \cref{lem_collisionprob_upperbound} can be found in \cref{apx_upperbound}.

\subsection{No-go for anticoncentration at low depth}
In this section, we study the properties of quantum circuit dynamics at sublogarithmic depth, which is defined as a limit where we fix $0\le c<1$ and scale depth as $d= O[(\log n)^c]$ while taking $n\to \infty$.
At this depth, there is still a notion of locality in the circuit because the lightcone $L_d(i)$ of each site cannot extend across the whole system in the large-$n$ limit.
We will prove that sampling from Haar random circuits on any parallel circuit architecture at sublogarithmic depth leads to a poorly anticoncentrated output distribution.
We consider the case of both noisy and noiseless circuits.

\begin{thm} \label{thm_no_anticonc}
Consider a Haar random circuit ensemble on any parallel architecture subject to Pauli noise with a uniform upper bound on the local noise rate $q_{\mu mi} \le q_\mu$ for all Pauli noise sectors $\mu$, $m$ and $i$.
Also, let $b= \min[q_x+q_y,q_y+q_z,q_z+q_x] < 1/2$ and $a' = \log 120 -2\log(1-2b)$.
If the depth of the circuit ensemble is sublogarithmic (i.e., satisfies $1 \leq d =o(\log n)$), then 
\begin{align}
\lim_{n\to \infty} \Pr_{\mathcal{B}}\left[p_{\oo}  < \frac{1}{2^n e^{ne^{-a'd}/4}} \right]  = 1.
\end{align}
\end{thm}
\begin{proof}
The strategy of the proof will be to show that a bound on the logarithm,
\begin{align}-\frac{1}{n} \log p_{\oo},
\end{align}
is sufficiently concentrated about its mean.
This can be thought of as the first step in proving a central limit theorem-like behavior for the log-output probability similar to the behavior of the free-energy in classical statistical mechanics.

To see this, we express the output probability $\po$ in terms of conditional probabilities:
\begin{align}
p(x_1=0,\ldots,x_n=0) = p(x_1=0)p(x_2=0|x_1=0) \ldots \nonumber\\
p(x_n=0 |x_1=0,\ldots,x_{n-1}=0),
\end{align}
which also holds when the $x_i$ are reordered by any permutation.  Taking the logarithm of both sides and summing over all permutation representations of the formula we arrive at the expression
\begin{align}
-\log p_{00\ldots 0} = -\frac{1}{n!} \sum_{\sigma \in S_n} \sum_i  \log p(x_{\sigma(i)} = 0 | x_{\sigma(1)}=0,\ldots \nonumber \\, x_{\sigma(i-1)}=0),
\end{align}
where $S_n$ is the permutation group on $(1,\ldots,n)$.
We rewrite $2 p(x_i=0|\{x_j=0\}_{j \in J_i}) = 1 +\langle{Z_i}\rangle_{J_i}$ to express the sum as
\begin{align}
& -\log p_{00\ldots 0} - n \log 2 = - \frac{1}{n!} \sum_{\sigma \in S_n} \sum_i \log( 1 +\langle Z_{\sigma(i)} \rangle_{\sigma(\{1,\ldots,i-1\} )}) \\
& \ge -\frac{1}{n!} \sum_{\sigma \in S_n} \sum_i \langle Z_{\sigma(i)} \rangle_{\sigma(\{1,\ldots,i-1\} )} + \frac{1}{4 n!} \sum_{\sigma \in S_n} \sum_i \langle Z_{\sigma(i)} \rangle_{\sigma(\{1,\ldots,i-1\} )}^2  \\
&= \frac{1}{n!} \sum_{\sigma \in S_n} A_{\sigma} + \frac{1}{4 n} \sum_{i,j} \frac{1}{(n-1)!} \sum_{\sigma_{ji}} \langle Z_j\rangle^2_{\sigma_{ji}(\{1,\ldots,i-1\})},
\end{align}
where we use the uniform lower bound $-\log( 1+x) \ge -x + x^2/4$ for $x \in[-1,1]$ and denote by $\sigma_{ji}$ the subset of permutations with $\sigma(i)=j$.
Now let $J_j^{(i)}$ run over all the subsets of $\{1,\ldots,n\}$ not containing $j$ of length $i-1$.
This set has size $\binom{n-1}{i-1}$.
For each $J_j^{(i)}$, the term $\langle Z_j \rangle^2_{J_{j}^{(i)}}$ appears multiple times in the above sum since there are multiple permutations $\tau$ satisfying $\tau(i)=j$ and $\{\tau(1),\tau(2)\ldots \tau(i-1)\}= \{\sigma(1),\sigma(2)\ldots \sigma(i-1)\}$.
The number of such permutations is $(i-1)! \times (n-i)!$.
Therefore, the above is equal to
\begin{align}
&\frac{1}{n!} \sum_{\sigma \in S_n} A_{\sigma} + \sum_{i,j} \frac{1}{4n \, \binom{n-1}{i-1} } \sum_{J_j^{(i)}} \langle Z_j \rangle^2_{J_{j}^{(i)}}  \\
&\ge \frac{1}{n!} \sum_{\sigma \in S_n} A_{\sigma} + \sum_j \sum_{i=1}^{n-L_j(d)+1} \frac{1}{4n} \frac{\binom{n-L_j(d)}{i-1}}{\binom{n-1}{i-1}} \langle Z_j \rangle^2,
\end{align}
where we denote $L_i(d)= |L_d\circ L_d^\dag(i)|$.
In the above, we have restricted the sum to only those subsets $J_j^{(i)}$ with $L_d\circ L_d^\dag(j) \cap J_j^{(i)} = \emptyset$.
This ensures conditional independence $p(x_j=0 |\{x_k = 0\}_{k\in J_j^{(i)}}) = p(x_j=0)$, which removes the conditional dependence on $\langle {Z_j} \rangle_{J_j^{(i)}}$.
The number of subsets $J_j^{(i)}$ such that there is no qubit in $L_d\circ L_d^\dag(j) \cap J_j^{(i)}$ is $\binom{n-L_j(d)}{i-1}$.
The sum over $i$ is truncated to $n-L_j(d)+1$ since for larger $i$, we will get back terms with conditional dependence.
Continuing, we use the hockey-stick relation $\sum_{i=1}^{n-L_j(d)+1} \binom{n-1}{n-L_j(d)-i+1} = \binom{n}{n-L_j(d)}$ to get
\begin{align}
-\log \po \geq n\log 2 +  \frac{1}{n!} \sum_{\sigma \in S_n} A_{\sigma} + \sum_j \frac{1}{4 L_j(d)} \langle Z_j \rangle^2 \\
\geq n\log 2 + \frac{1}{n!} \sum_{\sigma \in S_n} A_{\sigma} + \frac{1}{4 L_{\max}(d)}\sum_j  \langle Z_j \rangle^2 =: x, \label{eq_lowerbound}
\end{align}
where $L_{\max}(d)=\max_i L_i(d)$.

The proof of \cref{thm_lowerbound_vardist} provides the lower bound $\mathbb{E}_{\mathcal{B}} [\langle{Z_i}\rangle^2 ] \geq e^{-ad}$, while we also have $\mathbb{E}_{\mathcal{B}}[ A_\sigma ] = 0$.
As a result, we have a lower bound 
\begin{align}
\mathbb{E}_{\mathcal{B}}[-\log p_{00\ldots0}] \ge n \log 2 + n e^{-ad}/(4 L_{\max}(d)).
\end{align}

In \cref{lem_smallvariance}, we prove that the variance of the lower bound on $-\log p_{00\ldots0}$, denoted by $x$ in \cref{eq_lowerbound}, is $\sigma^2 \leq \sigma'^2 = 2n = o(n^2)$.
This implies that
\begin{align}
&\Pr_{\mathcal{B}}[\abs{x-\mathbb{E}_{\mathcal{B}}[x]} > k\sigma'] \leq \Pr[\abs{x-\mathbb{E}_{\mathcal{B}}[x]} > k\sigma] \leq \frac{1}{k^2}
\\ &\implies \Pr_{\mathcal{B}}[-x > -\mathbb{E}[x] + k\sigma'] \leq \frac{1}{k^2}
\\ &\implies  \Pr_{\mathcal{B}}[\log \po > -\mathbb{E}_{\mathcal{B}}[x] + k\sqrt{2n}] \leq \frac{1}{k^2}.
\end{align}
Since $L_\mathrm{max}(d) \leq 4^d$, we have $-\mathbb{E}_{\mathcal{B}}[x] \leq -n\log 2 -ne^{-a'd}/4$ with $a' = a + \log 4$, which means
\begin{align}
\Pr_{\mathcal{B}}\left[\po < 2^{-n} \exp[-\frac{ne^{-a'd}}{4}
+ k\sqrt{2n}]\right] &\geq 1-\frac{1}{k^2}.
\end{align}
Choose $k= \Theta(n^{0.01})$ so that
\begin{align}
\lim_{n\to \infty}\Pr_{\mathcal{B}}\left[\po < 2^{-n} \exp\left(-\frac{ne^{-a'd}}{4} + O(n^{0.51})\right)\right] = 1.
\end{align}
This means a lack of anticoncentration through $\lim_{n\to \infty} \Pr_{\mathcal{B}}[\po = o(2^{-n})]=1$ whenever ${ne^{-a'd}}  = \omega(n^{0.51})$, which is satisfied for any $d <  {0.49 \log n}/{a'}$, with $a' = \log 120 -2\log(1-2b)$, where $b=\min[q_x+q_y, q_y+q_z, q_z+q_x]$.
Thus we have established a lack of anticoncentration for any sublogarithmic depth $d$ subject to the proof of \cref{lem_smallvariance}.
\end{proof}
Before stating \cref{lem_smallvariance}, we state the following Corollary for noiseless circuits:
\begin{corll} \label{cor_noanticonc}
An ensemble of noiseless Haar random circuits on any parallel architecture at sublogarithmic depth $1 \leq d \leq o(\log n)$ is poorly anticoncentrated, satisfying
\begin{align}
\lim_{n\to \infty} \Pr_{\mathcal{B}}\left[p_{\oo}  < \frac{1}{2^n e^{n/(4\cdot 120^d)}} \right]  = 1.
\end{align}
\end{corll}
\begin{proof}
The proof for the noiseless case follows directly by taking $b=0$ in \cref{thm_no_anticonc}.
\end{proof}

We comment on the relation between this no-go result for anticoncentration at sublogarithmic depth and a related result of \citet{Dalzell2020}.
Dalzell et al.\ adopted a different definition of anticoncentration in terms of the expected collision probability $\mathbb{E}_\mathcal{B}[Z] \leq a/2^n$.
As mentioned earlier and justified in the proof of \cref{thm_noisy_anticonc}, a bound on the collision probability implies anticoncentration as defined in \cref{def_antconc}.
In the original context of quantum computational advantage proofs \cite{Aaronson2013a}, the \cref{def_antconc} of anticoncentration is more relevant.
Moreover, there are simple cases where this definition of anticoncentration is satisfied but there is not a good bound on the collision probability.
For example, consider a distribution where 1/2 of the probability mass is concentrated on a single outcome and the remaining 1/2 of the mass is equally distributed over all the other $2^n - 1$ outcomes.
The distribution satisfies \cref{def_antconc} by taking $\alpha =1/2$ and any $c < 1$.
The collision probability for this distribution is $Z = 1/4 + O(1/2^n) \gg a/2^n$.

\citet{Barak2020} and \citet{Dalzell2020} showed that random circuits in 1D and other parallel circuit architectures anticoncentrate (according to both \cref{def_antconc} and their definition in terms of the collision probability).
Dalzell et al.\ showed that logarithmic depth is also \emph{necessary} for the collision probability bound $\mathbb{E}_\mathcal{B}[Z] \leq a/2^n$, but this left open the possibility that sublogarithmic-depth random circuits anticoncentrate without a corresponding bound on the collision probability.
Our \cref{thm_no_anticonc} disproves this possibility and thus strengthens their no-go result at sublogarithmic depths.

\begin{lemma}\label{lem_smallvariance}
The variance $\sigma^2 := \mathbb{E}_\mathcal{B}[x^2]-\mathbb{E}_\mathcal{B}[x]^2$ of the lower bound $x$ in \cref{eq_lowerbound} satisfies $\sigma^2 \leq 2n$.
\end{lemma}

\subsection{Anticoncentration at large enough depth}
We now study anticoncentration properties of both noisy and noiseless circuits at depths logarithmic in $n$ or larger.
These results are established in terms of the collision probability defined in \cref{eq_collisionprob}, which is the second moment of the output probability.

\begin{thm}\label{thm_noisy_anticonc}
Haar random circuits with Pauli noise  anticoncentrate at least as fast as those without noise.
More specifically, we show that the output probability distribution for a Haar random circuit with Pauli noise is anticoncentrated if the noiseless output  satisfies $ 2^n\mathbb{E}_\mathcal{B} Z(U,\mathbb{I}) \leq c$ for some constant $c$.
\end{thm}
 
\begin{proof}
To prove the result we first need \cref{lemma_clifford_Z}, which shows that the collision probability for a noisy Clifford circuit increases upon removing the noise, i.e., $Z(U,\mathcal{E}) \leq Z(U,\mathbb{I})$.
The proof of the theorem then follows from the two-design property of the Clifford group through the inequalities
\begin{align} \label{eq:bound_ac1}
2^{-n} \le \mathbb{E}_\mathcal{B} Z(U,\mathcal{E})  &= \mathbb{E}_C Z(U,\mathcal{E}), \\ \label{eq:bound_ac2}
&\le\mathbb{E}_C Z(U,\mathbb{I}) = \mathbb{E}_\mathcal{B} Z(U,\mathbb{I}).
\end{align}
The first inequality in \cref{eq:bound_ac1} follows from a generic lower bound on $Z(U,\mathcal{E})$, while the second inequality in \cref{eq:bound_ac2} follows from \cref{lemma_clifford_Z}.

To prove that this bound on the collision probability implies anticoncentration, we can apply the Paley-Zygmund inequality for $0<\alpha < 1$
\begin{align}
\Pr \left[p_{00\ldots 0} \ge \frac{\alpha}{2^n} \right] &= \Pr \left[p_{00\ldots 0} \ge \alpha \mathbb{E}_\mathcal{B}[p_{00\ldots 0}] \right], \\ 
&\ge \frac{(1-\alpha)^2}{4^n\mathbb{E}_\mathcal{B}\left[p_{00\ldots0}^2 \right]} \ge \frac{(1-\alpha)^2}{c},
\end{align}
which converges to a number greater than $0$ in the limit $n\to \infty$.
Here, we used the fact that  
\begin{equation} 4^n\mathbb{E}_\mathcal{B}\left[ p_{00\ldots0}^2 \right] =2^n\mathbb{E}_\mathcal{B}\left[\sum_x p_x^2 \right]  = 2^n\mathbb{E}_\mathcal{B}[  Z(U,\mathcal{E})],
\end{equation}
and then applied \cref{eq:bound_ac2}.
\end{proof}
\Cref{thm_noisy_anticonc} implies that noisy Haar random circuits on  architectures that are reasonably well connected  anticoncentrate after $\Theta(\log n)$ depth since $2^n\mathbb{E}_\mathcal{B} Z(U,\mathbb{I})= 2 + o(1)$  after depth $\Theta(\log n)$ for such circuits \cite{Dalzell2020}\footnote{See \citet{Dalzell2020} for a sufficient set of connectivity conditions.}.
\Cref{thm_no_anticonc} also rules out anticoncentration at any sub-logarithmic depth for a Haar random circuit with Pauli noise.

We now present the  lemma for noisy Clifford circuits:
\begin{lemma}\label{lemma_clifford_Z}
 The collision probability  decreases when adding Pauli noise to a noiseless Clifford circuit, i.e., 
\begin{align}
Z(U,\mathcal{E}) \le Z(U,\mathbb{I}).
\end{align}
\end{lemma}

\begin{proof}

The defining properties of Clifford circuits is that they send Pauli group elements to Pauli group elements.  Another result we need is that the composition of two Pauli noise channels is also a Pauli noise channel.  We can use these two facts to move all the noise channels past the Clifford gates to arrive at the expression
\begin{align}\mathcal{N}_d(\ket{0}\bra{0}) = (\mathcal{E}_U \circ \mathcal{C}) (\ketbra{\oo}),
\end{align}
where $\mathcal{E}_U$ is a new Pauli noise channel with a modified distribution $q(E)$ and $\mathcal{C} = \mathcal{C}_d \circ \cdots \circ \mathcal{C}_1$ is the noiseless Clifford circuit.  Now we are left to bound the expression
\begin{align}
Z(U,\mathcal{E}) &= \sum_x \left| \sum_{E \in \mathcal{P}} q(E)  \bra{x} E \left[  \prod_{i=1}^n \frac{\mathbb{I} + g_i}{2} \right] E^\dag \ket{x} \right|^2, \nonumber \\
&= \sum_x \left| \sum_{\vec{s}} q_{\vec{s}} \, \bra{x}  \prod_{i=1}^n \frac{\mathbb{I} + (-1)^{s_i} g_i }{2}  \ket{x} \right|^2 \label{eqn:zn1},
\end{align}
where $g_i = U_d \cdots U_1 Z_i U_1^\dag \cdots U_d^\dag$ are stabilizer generators for the evolved initial state under the noiseless circuit and $Z_i$ is the Pauli $z$-operator on site $i$.  In the second equality, we have organized the sum into syndrome classes defined by the anticommutation pattern $\vec{s}$ of the Pauli group elements $E$ with the stabilizer generating set $\{ g_i \}$ using the definition
\begin{align}q_{\vec{s}} := \sum_{E \in \mathcal{P} \, s.t.\, (-1)^{s_i} = [[E,g_i]]} q(E).
\end{align}
Here, $[[A,B]] = \Tr[A B A^{-1} B^{-1}]/2^n$ is the scalar commutator. 
To bound Eq.~\eqref{eqn:zn1}, we first use properties of stabilizer states to evaluate the measurement probabilities as
\begin{align}\bra{x}  \prod_{i=1}^n \frac{\mathbb{I} + (-1)^{s_i} g_i }{2} \ket{x} = p_{\max}(U,\mathbb{I}) f(x,\vec{s},U),
\end{align}
where $p_{\max}(U,\mathbb{I}) = \max_x p_x(U,\mathbb{I})$, $p_x(U,\mathbb{I}):= \abs{\bra{x} U \ket{\oo}}^2$, and $f(x,\vec{s},U)$ is either 0 or 1 and satisfies the sum rule $\sum_x f(x,\vec{s},U) = p_{\max}^{-1}(U,\mathbb{I})$ for every $\vec{s}$.  As a result,
\begin{align}
Z(U,\mathcal{E}) &= p^2_{\max}(U,\mathbb{I}) \sum_{\vec{s},\vec{\ell}} q_{\vec{s}} \, q_{\vec{\ell}} \sum_x f(x,\vec{s},U) f(x,\vec{\ell},U), \\
&\le p_{\max}^2(U,\mathbb{I}) \sum_{\vec{s},\vec{\ell}} q_{\vec{s}} \, q_{\vec{\ell}} \sum_x f(x,\vec{s},U), \\
&= p_{\max}(U,\mathbb{I}) = Z(U,\mathbb{I}), 
\end{align}
where we used the fact that $f(x,\vec{s},U) f(x,\vec{\ell},U) \le f(x,\vec{s},U)$, $\sum_{\vec{s}} q_{\vec{s}} = 1$, and the property of Clifford circuits that all nonzero probabilities are equal.
\end{proof}

\section*{Acknowledgments}
We thank Alex Dalzell for helpful and inspiring discussions. We thank Igor Boettcher and Grace Sommers for helpful comments on the manuscript. M.J.G., A.V.G., and P.N.~acknowledge support  from the National Science Foundation (QLCI grant OMA-2120757). 
A.D., A.V.G., P.N., and O.S.~acknowledge funding by DoE QSA, DoE ASCR Accelerated Research in Quantum Computing program (award No.~DE-SC0020312), DoE ASCR Quantum Testbed Pathfinder program (award No.~DE-SC0019040), NSF PFCQC program, AFOSR, U.S.~Department of Energy Award No.~DE-SC0019449, ARO MURI, AFOSR MURI, and DARPA SAVaNT ADVENT. 
A.\,D.\, also acknowledges support from the National Science Foundation RAISE-TAQS 1839204.
B.F.~acknowledges support from AFOSR (YIP number FA9550-18-1-0148 and
FA9550-21-1-0008). This material is based upon work partially
supported by the National Science Foundation under Grant CCF-2044923
(CAREER) and by the U.S.~Department of Energy, Office of Science,
National Quantum Information Science Research Centers.
The Institute for Quantum Information and Matter is an NSF Physics Frontiers Center PHY-1733907.

\onecolumngrid
\appendix

\section{Proof of \cref{lem_collisionprob_upperbound}} \label{apx_upperbound}
Here, we introduce the concepts that go into proving \cref{lem_collisionprob_upperbound}. 

\medskip

If $X$ is a random variable in $\{0,1\}^n$ denoting the measurement outcome of $n$ qubits after a unitary $U \sim \mathcal{B}$, the collision probability of the random-circuit architecture is defined as
\begin{equation}
Z = \mathbb{E}_{\mathcal{B}}\left[\sum_{x \in \{0, 1\}^n} \text{Pr}(X=x)^2\right] = \mathbb{E}_{\mathcal{B}}\left[\sum_{x \in \{0, 1\}^n} p_U(x)^2 \right],
\end{equation}
where $p_U(x)$ is the probability that the measurement result is $x$.
If there is at least one gate for each qubit in a parallel circuit architecture with Haar-random gates, all measurement outcomes are equally random and, thus, there is a symmetry over them.
\begin{equation}
    Z = \mathbb{E}_{\mathcal{B}}\left[\sum_{x \in \{0, 1\}^n} p_U(x)^2 \right] = 2^n\mathbb{E}_{\mathcal{B}}\left[ p_U({0}^n)^2\right],
\end{equation}
where ${0}^n$ denotes the state $\ket{00\ldots 0}$.
Assuming that the input state is also ${0}^n$, the probability of measuring ${0}^n$ after the circuit is given by $\Tr(\ketbra{{0}^n} U\ket{0^n}\bra{0^n}U^\dagger)$. To get the second moment of the probability distribution, we consider two copies of the circuit acting on two copies of the input state. Since the trace obeys $\Tr(A\otimes B)=\Tr(A)\Tr(B)$, we get
\begin{align}
Z = 2^n\mathbb{E}_{\mathcal{B}}\left[ p_U({0})^2\right]  &= 2^n \mathbb{E}_\mathcal{B}\Tr[\left(\ketbra{{0}^n}\right)^{\otimes 2} U^{\otimes 2}\left(\ketbra{0^n}\right)^{\otimes 2}(U^\dagger)^{\otimes 2}], \nonumber \\
&=  2^n \Tr\left[\left(\ketbra{0^n}\right)^{\otimes 2} \mathbb{E}_{\mathcal{B}}\left[U^{\otimes 2}\left(\ketbra{0^n}\right)^{\otimes 2}(U^\dagger)^{\otimes 2}\right]\right]. \label{eq:def-collision-probability}
\end{align}
For convenience, we denote the two-copy, Haar-averaged channel over $k$ qubits as $M_{U_k}$:
\begin{equation}
    M_{U_k}[\rho] = \mathbb{E}_{\mathcal{B}}\left[{U_k}^{\otimes 2}\rho({U_k}^\dagger)^{\otimes 2}\right],
\end{equation}
where $U_k$ is a unitary acting on $k$ qubits. To study noisy evolution, we define a dephasing noise of strength $q$ given by the noise channel
\begin{equation}
    \mathcal{E}[\rho] = (1-q)\rho + q Z \rho Z.
\end{equation}

\setcounter{lemma}{1}
\begin{lemma}
For a noisy Haar random circuit of depth $d$ on any parallel circuit architecture with heralded dephasing noise at rate $p$ with the dephasing parameter $q$, we have the upper bound on the expected collision probability
\begin{equation}
\mathbb{E}_{\mathcal{B}}[Z] = \mathbb{E}_{\mathcal{B}}\Big[ \sum_x p_x^2 \Big]\leq 
2^{-n} \exp[\frac{n}{3}e^{-\gamma pd}],
\end{equation}
where $\gamma=8 q (1-q)/3.$
\end{lemma}
\setcounter{lemma}{4}
\begin{proof} 
It is convenient to separate out the average over random locations of the dephasing events in the heralded dephasing model from the ensemble $\mathcal{B}$ (an ensemble of both gates and noise locations). We denote the ensemble of gates for a fixed set of locations $L$ by $\mathcal{B}_L$ and averages over the noise locations by $\mathbb{E}_\mathcal{L}$.
We will show in Lemma \ref{lem:modified-circuit} that given a random circuit ensemble $\mathcal{B}_L$ on a parallel circuit architecture  with heralded dephasing noise and a fixed set of locations $L$, there exists another circuit ensemble $\mathcal{B}_L'$, with the gates drawn at random independently of $L$, composed solely of  single-qubit gates and SWAP gates with an average collision probability $\mathbb{E}_{\mathcal{B}_L'}[Z]$ greater than or equal to the average collision probability of the original circuit $\mathbb{E}_{\mathcal{B}_L}[Z]$.  
Note, for every circuit in the new ensemble $C \in \mathcal{B}_L'$, we can append a network of SWAP gates to $C$ to return all qubits to their original positions.
Adding these SWAP gates does not change the collision probability, since these gates only permute the support of the final probability distribution. We further break up the ensemble $\mathcal{B}_L'$ into an ensemble of single-qubit Haar random gates $\mathcal{B}_{1}'$ and the random SWAP network $\mathcal{B}_{\rm SWAP}'$.  
The distribution of gates in the ensemble $\mathcal{B}_1' \cup \mathcal{B}_{\rm SWAP}'$ is defined by taking every two-qubit gate for the circuits in $\mathcal{B}_L$ and replacing it with a SWAP gate  or identity gate with equal probability on those sites, followed by Haar-random single-site gates on the two qubits.
The joint distribution over SWAP networks and single-site gates is also conditionally independent, allowing us to commute averages over single-site gates, SWAP networks, and noise locations with each other.

Fixing a realization of SWAP gates and noise locations, we can then follow the path of a qubit, count the total number of dephasing events in that path and merge consecutive single-qubit gates without intervening noise locations. Since we are working with a parallel circuit architecture, there is never a case of consecutive dephasing events (there is always a single-qubit gate after each dephasing event).
Let $t_i$ be the number of dephasing events on the path of qubit $i$. For the heralded dephasing model, we can write the random variable $t_i$ as a sum
\begin{equation}
    t_i = \sum_{j=1}^d x_{ij},
\end{equation}
where $x_{ij}\in\{0,1\}$ are independent, identically distributed Bernoulli random variables for each $i$ and $j$ with parameter $p$. 
We have $\mathbb{E}_{\mathcal{L}}[x_{ij}]=p$ and $\Pr_\mathcal{L}[x_{ij} = a, x_{kl} =b] = \Pr_\mathcal{L}[x_{ij} = a]\Pr_\mathcal{L}[x_{kl} =b]$ for every $ij \ne kl$.

After averaging over $\mathcal{B}_1'$ using Lemma \ref{lem:single-qubit}, the final two-copy circuit-averaged state is given by
\begin{equation}
\bigotimes_{i=1}^n \left(M_{U_1} \circ \underbrace{(\mathcal{E}\circ M_{U_1})   \cdots    \circ (\mathcal{E}\circ M_{U_1})}_{t_i}\right) [\ket{0^n}\bra{0^n}^{\otimes 2}] \\
=  \bigotimes_{i=1}^n \left[\left(\frac{1}{12}(3-\beta^{t_i})\right)I + \frac{1}{6}\beta^{t_i}S\right],
\end{equation}
where $\beta=1-8q(1-q)/3$, and $I$ and $S$ are the $4 \times 4$ identity and SWAP operators, respectively. Using \cref{eq:def-collision-probability} and noting that $\Tr(I\ketbra{0}^{\otimes 2}) = \Tr(S\ketbra{0}^{\otimes 2})=1$, the average collision probability for a fixed SWAP network and set of noise locations equals
\begin{equation}
\mathbb{E}_{\mathcal{B}_{1}'}[Z] = 2^n\prod_{i=1}^n \left[\frac{1}{12}(3-\beta^{t_i}) + \frac{1}{6}\beta^{t_i}\right] = 2^n\prod_{i=1}^n \frac{1}{2^2}\left[1+\frac{1}{3}\beta^{t_i}\right] = \frac{1}{2^n}\prod_{i=1}^n \left[1+\frac{1}{3}\beta^{t_i}\right].
\end{equation}
We now  average over the noise locations using our assumption that the noise locations are uncorrelated with each other and the realization of the SWAP network.
\begin{equation}
    \mathbb{E}_{\mathcal{L}}\mathbb{E}_{\mathcal{B}_{1}'}[Z] = \mathbb{E}_{\substack{t_i \\ 1\leq i \leq n}}\left[{\mathbb{E}_{\mathcal{B}_{1}'}}[Z]\right]  = \frac{1}{2^n}\prod_{i=1}^n \left[1+\mathbb{E}_{t_i}\left[\frac{1}{3}\beta^{t_i}\right]\right] = \frac{1}{2^n}\prod_{i=1}^n \left[1+\frac{1}{3}\prod_{j=1}^d\mathbb{E}_{x_{ij}}\left[\beta^{x_{ij}}\right]\right].
    \label{eq:location-average}
\end{equation}
Using the fact that $x_{ij}$ are Bernoulli random variables, we can compute the expectation $\mathbb{E}_{x_{ij}}\left[\beta^{x_{ij}}\right]$ as
\begin{equation}
    \mathbb{E}_{x_{ij}}\left[\beta^{x_{ij}}\right] = p\beta^1 + (1-p)\beta^0 = p\beta + (1-p) = 1-p(1-\beta) = 1-p\gamma,
\end{equation}
where we have defined $\gamma = 1-\beta$. Inserting this expectation in \eqref{eq:location-average} and using  \cref{lem:modified-circuit}, we get a bound on the average collision probability 
\begin{equation}
     \mathbb{E}_{\mathcal{B}}[Z] \le \mathbb{E}_{\mathcal{B}'}[Z]=\mathbb{E}_{\mathcal{B}'_{\rm SWAP}}\mathbb{E}_{\mathcal{L}}\mathbb{E}_{ \mathcal{B}_{1'}}[Z] 
     = \frac{1}{2^n}\left[1+\frac{1}{3}(1-p\gamma)^d\right]^n  \le \frac{1}{2^n}\bigg[ 1+\frac{1}{3}e^{-\gamma p d} \bigg]^n \le 2^{-n} \exp\bigg[ \frac{n}{3} e^{-\gamma p d} \bigg].
\end{equation}
Here, we used $\mathbb{E}_{\mathcal{B}}[Z]=\mathbb{E}_\mathcal{L} \mathbb{E}_{\mathcal{B}_L}[Z] \le   \mathbb{E}_{\mathcal{B}'}[Z] := \mathbb{E}_\mathcal{L} \mathbb{E}_{\mathcal{B}_L'}[Z]  $ in applying \cref{lem:modified-circuit} in the first inequality.  We also used the fact that, for any $x>0$, $1-x \le e^{-x}$ in the second inequality and $1+x \le e^x $ in the third inequality.   

\end{proof}

\begin{lemma} Consider a random circuit consisting of $k$ dephasing error channels of strength $q$ sandwiched between $k+1$ single-qubit Haar-random gates (denoted by $U_1$). When this circuit acts on two copies of a single-qubit, the circuit-averaged state is given by
\begin{equation}
    M_{U_1} \circ \underbrace{(\mathcal{E}\circ M_{U_1})   \cdots    \circ (\mathcal{E}\circ M_{U_1})}_{k}  [\ketbra{0}^{\otimes 2}] = \frac{1}{12}(3-\beta^k)I + \frac{1}{6}\beta^{k}S,
\end{equation}
where $\beta=1-\frac{8}{3}q(1-q)$, and $I$ and $S$ are the $4 \times 4$ identity and SWAP operators, respectively.
\label{lem:single-qubit}
\end{lemma}
\begin{proof} We first make an observation that $M_{U_1}[\rho] = M_{U_1}\circ M_{U_1}[\rho]$, that is, one can split a Haar-random gate into two Haar-random gates without changing the statistics. In the circuit described above, leaving the two terminal unitary gates intact, we split all inner gates into two. This lets us treat the circuit as a repeating sequence of  $n$ units of the composite channel $\tilde{M}_{U_1, \mathcal{E}} = M_{U_1}\circ \mathcal{E} \circ M_{U_1}$.

From Ref.~\cite{Dalzell2020}, we know the following:
\begin{equation}
    M_{U_1}[\sigma] = \frac{1}{3} \left(\Tr(\sigma) - 2^{-1}\Tr(\sigma S)\right)I +  \frac{1}{3} \left(\Tr(\sigma S) - 2^{-1}\Tr(\sigma)\right)S.
    \label{eq:single-qubit-channel}
\end{equation}
If we follow this gate by a dephasing error channel, we get
\begin{align*}
     \mathcal{E} \circ  M_{U_1}[\sigma] &=  \frac{1}{3} \left(\Tr(\sigma) - 2^{-1}\Tr(\sigma S)\right)\mathcal{E}[I] +  \frac{1}{3} \left(\Tr(\sigma S) - 2^{-1}\Tr(\sigma)\right)\mathcal{E}[S], \\
      &= \frac{1}{3} \left(\Tr(\sigma) - 2^{-1}\Tr(\sigma S)\right)I +  \frac{1}{3} \left(\Tr(\sigma S) - 2^{-1}\Tr(\sigma)\right)((1-q)^2 S + 2q(1-q) (ZI)S(ZI)  + q^2(ZZ)S(ZZ)), \\ 
      &= \frac{1}{3} \left(\Tr(\sigma) - 2^{-1}\Tr(\sigma S)\right)I +  \frac{1}{3} \left(\Tr(\sigma S) - 2^{-1}\Tr(\sigma)\right)\left(((1-q)^2 + q^2)S + 2q(1-q) (ZI)S(ZI)\right).
\end{align*}
We follow this channel by another single-qubit random gate to finish the composite block. First we observe that $M_{U_1}[I] = I$ and $M_{U_1}[S]=S$. Similarly using \eqref{eq:single-qubit-channel} together with the fact that $\Tr[S]=2, \Tr[(IZ)S(IZ)S]=0$,
\begin{align*}
    M_{U_1}[(IZ)S(IZ)] &= \frac{2}{3}I -  \frac{1}{3}S.
\end{align*}
The composite channel thus gives
\begin{align}
    \tilde{M}_{U_1, \mathcal{E}}[\sigma] &= \frac{1}{3} \left(\Tr(\sigma) - 2^{-1}\Tr(\sigma S)\right) M_{U_1}[I] +  \frac{1}{3} \left(\Tr(\sigma S) - 2^{-1}\Tr(\sigma)\right)\left(((1-q)^2 + q^2)M_{U_1}[S] + 2q(1-q) M_{U_1}[(ZI)S(ZI)]\right), \nonumber \\
    &=  \frac{1}{3} \left(\Tr(\sigma) - 2^{-1}\Tr(\sigma S)\right)I +  \frac{1}{3} \left(\Tr(\sigma S) - 2^{-1}\Tr(\sigma)\right)\left(((1-q)^2 + q^2)S + 2q(1-q)\left(\frac{2}{3}I -  \frac{1}{3}S \right) \right), \nonumber \\
    &= \frac{1}{3} \left(\Tr(\sigma) - 2^{-1}\Tr(\sigma S) + \underbrace{\frac{4}{3}q(1-q)}_{\alpha}\left(\Tr(\sigma S) - 2^{-1}\Tr(\sigma )\right)\right)I +  \frac{1}{3} \left(\Tr(\sigma S) - 2^{-1}\Tr(\sigma )\right)\underbrace{\left(1-\frac{8}{3}q(1-q)\right)}_{\beta}S, \nonumber \\
    &= \frac{1}{3} \left(\Tr(\sigma) - 2^{-1}\Tr(\sigma S) + \alpha\left(\Tr(\sigma S) - 2^{-1}\Tr(\sigma )\right)\right)I +  \frac{1}{3} \left(\Tr(\sigma S) - 2^{-1}\Tr(\sigma)\right)\beta S.
    \label{eq:composite-channel}
\end{align}
The composite sum returns a state of the form $a I + b S$. Acting on this state with other composite blocks only changes the coefficients $a$ and $b$. In fact, we can work out exactly how $a$ and $b$ change after each block. Knowing that $\Tr(aI+bS) = 4a + 2b$ and $\Tr(S(aI + bS)) = 2a+4b$, we get
\begin{align}
    \tilde{M}_{U_1, \mathcal{E}}[aI + bS] &= \frac{1}{3} \left(3a+ \alpha(3b)\right)I +  \frac{1}{3}(3b)\beta S = (a + \alpha b)I + (\beta b)S.
    \label{eq:composite-channel-coefficients}
\end{align}
The first composite acts on the state $\ketbra{0}^{\otimes 2}$. Knowing that $\Tr[\ketbra{0}^{\otimes 2}]= \Tr(S\ketbra{0}^{\otimes 2})=1$, and using \cref{eq:composite-channel},
\begin{align}
\tilde{M}_{U_1, \mathcal{E}}[\ketbra{0}^{\otimes 2}] &=  \frac{1}{3} \left(1 - 2^{-1} + \alpha\left(1 - 2^{-1}\right)\right)I +  \frac{1}{3} \left(1 - 2^{-1}\right)\beta S  = \frac{1}{6} \left(1 + \alpha \right)I +  \frac{1}{6}\beta S.
\end{align}
We take this state and apply another $k-1$ composite channels (since there are $k$ in total). We can calculate the final state recursively using \cref{eq:composite-channel-coefficients}:
\begin{equation}
   \tilde{M}_{U_1, \mathcal{E}} \circ \cdots \circ \tilde{M}_{U_1, \mathcal{E}}[\ketbra{0}^{\otimes 2}] = \frac{1}{6}\left[1+\alpha+ \alpha\sum_{i=1}^{k-1} \beta^i\right]I + \frac{1}{6}\beta^{k}S  = \frac{1}{12}(3-\beta^k) I + \frac{1}{6}\beta^k S.
\end{equation}
This proves the lemma.
\end{proof}

It now remains to prove the following lemma, which lets us put a bound on the average collision probability in \cref{lem_collisionprob_upperbound}.

\begin{lemma}
Consider a random quantum circuit ensemble on a parallel architecture, $\mathcal{B}_L$, with Haar-random two-qubit gates and heralded dephasing noise with a fixed set of noise locations $L$.
There is a procedure to obtain another circuit ensemble $\mathcal{B}_L'$ with gates drawn randomly independently of $L$, composed solely of noisy single-qubit channels and SWAP gates, with an equal or higher average collision probability, i.e., $\mathbb{E}_{\mathcal{B}_L}[Z] \leq \mathbb{E}_{\mathcal{B}_L'}[Z]$.
\label{lem:modified-circuit}
\end{lemma}
\begin{proof}
Note that without loss of generality, we can assume that every two-qubit Haar-random gate is preceded by Haar-random single-qubit gates on both incoming lines, since the two situations correspond to the same ensemble.
When the noise is heralded, the circuit consists of three kinds of two-qubit gates:
\begin{enumerate}
\item Type A, where the two-qubit gate is noiseless.
\item Type B, where one of the two outgoing legs of the gate undergoes dephasing, followed by a single-qubit random gate.
\item Type C, where both outgoing legs undergo dephasing followed by single-qubit gates on both legs.
\end{enumerate}

We will analyze each of these types separately and start with a brief review of the methods of Ref.~\cite{Dalzell2020} for noiseless random circuits.
In Ref.~\cite{Dalzell2020}, it was shown that taking a two copies of an $n$-qubit state $\ketbra{0^n}^{\otimes 2}$, acted on by a Haar random circuit $U\otimes U$, and averaged over the unitaries leads to a density matrix that can be represented by a linear combination of length-$n$ configurations in $\{I, S\}^n$ where $I$ and $S$ are the $4\times 4$ identity and SWAP matrices. Any two-qubit unitary gate takes a linear combination to another linear combination. More precisely, for $\vec{\gamma} \in \{I, S\}^{ n}$, a two-qubit unitary channel $M_{U_2}$, acting on qubits $i$ and $j$, transforms it to
\begin{equation}
 M_{U_2}[\vec{\gamma}] = M_{U_2}\left[ \bigotimes_{a=1}^n \gamma_a\right] = \sum_{\vec{\nu} \in \{I, S\}^{ n}} M_{U_2}^{\vec{\gamma},\vec{\nu}} \bigotimes_{b=1}^n\nu_b =: \sum_{\vec{\nu} \in \{I, S\}} M_{U_2}^{\vec{\gamma},\vec{\nu}} \vec{\nu},
\end{equation}
where $M_{U_2}^{\vec{\gamma},\vec{\nu}}$ are matrix elements determined by qubit locations $i,j$
\begin{equation}
M_{U_2}^{\vec{\gamma},\vec{\nu}} =  \begin{cases} 1 & \text{ if } \gamma_i = \gamma_j\text{ and }\vec{\gamma}=\vec{\nu}, \\ 
2/5 & \text{ if } \gamma_i \neq \gamma_j \text{ and } \nu_i=\nu_j \text{ and } \gamma_k = \nu_k ~\forall k \in [n]/[i,j], \\ 
    0 & \text{ otherwise}.
\end{cases}
    \label{eq:transition-matrix}
\end{equation}
Therefore, a state can be represented as a linear combination of trajectories of the configuration strings, with each trajectory weighted according to \cref{eq:transition-matrix}. Furthermore, since $\Tr(\vec{\gamma} \ketbra{0^n}^{\otimes 2})=1$ for each $\vec{\gamma} \in \{I, S\}^n$, the collision probability can be, similarly, written as a sum over weighted trajectories. More precisely, the average collision probability of a circuit with $s$ gates,
\begin{equation}
    \mathbb{E}_{\mathcal{B}_L}[Z_s] = \frac{1}{3^n} \sum_{\gamma \in \{I, S\}^{n\times s}}\prod_{t=1}^{s-1}M_{U_2}^{\vec{\gamma}_t, \vec{\gamma}_{t+1}} = \frac{1}{3^n} \sum_{\gamma \in \{I, S\}^{n\times s}}\prod_{t=1}^{s-1} \text{wt}(\gamma).
    \label{eq:collision-weights-decomposition}
\end{equation}
In the above, the factor $1/3^n$ comes from the fact that after the first layer of Haar-random single-qubit gates, the Haar-averaged two-copy state is given by $\frac{1}{6^n} \sum_{\vec{\gamma}\in \{I,S\}^n} \vec{\gamma}$, the uniform mixture of all configurations in $\{I,S\}^n$.
Also, the weight $ \text{wt}(\gamma)$ of a configuration is defined as the product of the matrix elements $M_{U_2}^{\vec{\gamma}_t, \vec{\gamma}_{t+1}}$.

Now, we modify this construction to account for noise.
We add one more gate of Type $A$, $B$ or $C$ to this circuit. Since all three types are two-qubit gates, we let $[i,j]$ denote the qubits the gate acts on. We can isolate the qubits $[i,j]$ from the decomposition in \cref{eq:collision-weights-decomposition} as follows:
\begin{equation}
    \mathbb{E}_{\mathcal{B}_L}[Z_s] = \frac{1}{3^n}\left[ \sum_{\substack{\gamma \in \{I, S\}^{ns} \\ \vec{\gamma}^s_{ij}=II}} \text{wt}(\gamma) +  \sum_{\substack{\gamma \in \{I, S\}^{ns} \\ \vec{\gamma}^s_{ij}=IS}} \text{wt}(\gamma) + \sum_{\substack{\gamma \in \{I, S\}^{ns)} \\ \vec{\gamma}^s_{ij}=SI}} \text{wt}(\gamma)+\sum_{\substack{\gamma \in \{I, S\}^{ns} \\ \vec{\gamma}^s_{ij}=SS}} \text{wt}(\gamma)\right].
\end{equation}
\paragraph{Type A:} When we add a noiseless two-qubit gate, the bit-strings transform according to \cref{eq:transition-matrix}. Zooming on qubits $i$ and $j$, the trajectories evolve as follows:
\begin{equation}
   M_{U_2}[II] =II \qquad M_{U_2}[SS] =SS  \qquad M_{U_2}[IS, SI] =\frac{2}{5}(II + SS).
\end{equation}
The trajectories for which $\vec{\gamma}_{ij}^s \in \{II, SS\}$  have their weights unchanged. The trajectories for which $\vec{\gamma}_{ij}^s \in \{IS, SI\}$ have their weights changed by $4/5$ (the trajectory splits two ways, each weighted by $2/5$).
\begin{equation}
       \mathbb{E}_{\mathcal{B}_L}[Z_{s+1}] = \frac{1}{3^n} \left[ \sum_{\substack{\gamma \in \{I, S\}^{n(s+1)} \\ \vec{\gamma}^s_{ij}=II}} \text{wt}(\gamma) +  \frac{4}{5}\sum_{\substack{\gamma \in \{I, S\}^{n(s+1)} \\ \vec{\gamma}^s_{ij}=IS}} \text{wt}(\gamma) + \frac{4}{5}\sum_{\substack{\gamma \in \{I, S\}^{n(s+1)} \\ \vec{\gamma}^s_{ij}=SI}} \text{wt}(\gamma)+\sum_{\substack{\gamma \in \{I, S\}^{n(s+1)} \\ \vec{\gamma}^s_{ij}=SS}} \text{wt}(\gamma)\right].
\end{equation}
If, instead, we consider a modified random circuit where the two-qubit gate consists of a SWAP gate or identity with probability 1/2 followed by Haar random single-qubit gates, all the trajectories retain their original weights since the collision probability is invariant under a SWAP gate and $M_{U_1}[I]=I$ and $M_{U_1}[S]=S$. Denoting the locally modified circuit ensemble with the same set of noise locations by $\mathcal{B}_L'$, we have, in both cases,
\begin{equation}
    \mathbb{E}_{\mathcal{B}_L'}[Z_{s+1}] = \mathbb{E}_{\mathcal{B}_L}[Z_s] \implies \mathbb{E}_{\mathcal{B}_L'}[Z_{s+1}] > \mathbb{E}_{\mathcal{B}_L}[Z_{s+1}].
\end{equation}
\begin{align}
\begin{tikzcd}[transparent,row sep=0.7em]
& \gate[5]{C} & \qw & \qw \\
& \qw & \qw & \qw \\
& \qw& \gate[2]{U_2} & \qw\\
&\qw & & \qw \\
& \qw & \qw & \qw
\end{tikzcd}
\to
\frac{1}{2}
\begin{tikzcd}[transparent,row sep=0.7em]
& \gate[5]{C} & \qw & \qw \\
& \qw & \qw & \qw \\
& \qw& \gate[]{U_1} & \qw\\
&\qw & \gate[]{U_1} & \qw \\
& \qw & \qw & \qw
\end{tikzcd}
+
\frac{1}{2}
\begin{tikzcd}[transparent,row sep=0.7em]
& \gate[5]{C} & \qw & \qw & \qw \\
& \qw & \qw & \qw & \qw\\
& \qw & \gate[swap]{} & \gate{U_1} & \qw\\
&\qw & & \gate{U_1} &\qw \\
& \qw & \qw &\qw & \qw
\end{tikzcd}
\end{align}


\paragraph{Type B:} The gate of Type B has a noiseless two-qubit gate followed by dephasing on one of the outgoing legs. The dephasing is also followed by a single-qubit random gate. To simplify things, we first understand the effect of the channel $M_{U_1}\otimes \mathcal{E}$ on $I$ and $S$. Of course, $M_{U_1}\circ \mathcal{E}[I]=I$, since neither the error nor the random gate has any effect on the identity matrix. However, for $S$, we get
\begin{align}
    M_{U_1}^{(i)}\circ \mathcal{E}^{(i)}[S] &= M_{U_1}[(1-q)^2 S + q(1-q)(IZ)S(IZ) + q(1-q)(ZI)S(ZI) + q^2 S], \\
    &= \alpha I + \beta S,
\end{align}
with $\alpha = 4q(1-q)/3, \beta = 1-8q(1-q)/3$. Without losing generality, we assume that the dephasing happens on gate $i$, and the dephasing channel is denoted by $\mathcal{E}$. We now tabulate the effect of this composite channel:
\begin{align}
    M_{U_1}^{(i)}\circ \mathcal{E}^{(i)}\circ M_{U_2}[II] &= II, \\ 
    M_{U_1}^{(i)}\circ \mathcal{E}^{(i)}\circ M_{U_2}[SS] &=  \alpha IS + \beta SS, \\
    M_{U_1}^{(i)}\circ \mathcal{E}^{(i)}\circ M_{U_2}[IS, SI] &= M_{U_1}^{(i)}\left[\frac{2}{5}(II + SS)\right] = \frac{2}{5}(II + \alpha IS + \beta SS).
\end{align}
The average collision probability of the new circuit is given by
\begin{equation}
       \mathbb{E}_{\mathcal{B}_L}[Z_{s+1}] = \frac{1}{3^n} \left[ \sum_{\substack{ \vec{\gamma}^s_{ij}=II}} \text{wt}(\gamma) +  \frac{2}{5}(1+\alpha + \beta)\left(\sum_{\substack{ \vec{\gamma}^s_{ij}=IS}} \text{wt}(\gamma) + \sum_{\substack{ \vec{\gamma}^s_{ij}=SI}} \text{wt}(\gamma)\right)+(\alpha + \beta) \sum_{\substack{\vec{\gamma}^s_{ij}=SS}} \text{wt}(\gamma)\right].
\end{equation}
Using the same locally modified circuit ensemble as above, we obtain
\begin{align}
    M_{U_1}^{(i)}\circ \mathcal{E}^{(i)}\circ M_{U_2} &\to  \frac{1}{2}M_{U_1}^{(i)}\circ \mathcal{E}^{(i)}\circ M_{U_1}^{(i)} \circ M_{U_1}^{(j)} +  \frac{1}{2}M_{U_1}^{(i)}\circ \mathcal{E}^{(i)}\circ M_{U_1}^{(i)}\circ M_{U_1}^{(j)} \circ \text{SWAP} \\
    & \quad =  \frac{1}{2}M_{U_1}^{(i)}\circ \mathcal{E}^{(i)}\circ M_{U_1}^{(i)} \circ M_{U_1}^{(j)} +  \frac{1}{2}M_{U_1}^{(j)}\circ \mathcal{E}^{(j)}\circ M_{U_1}^{(j)} \circ M_{U_1}^{(i)}.
\end{align}
Under this new composite channel, the bit-strings evolve as follows:
\begin{equation}
    II \to II, \qquad SS \to \beta SS + \frac{\alpha}{2}(IS + SI), \qquad IS \to \frac{1}{2}(IS + \alpha II + \beta IS), \qquad SI \to \frac{1}{2}(\alpha II + \beta SI + SI).
\end{equation}
The collision probability of the modified circuit is given by
\begin{equation}
       \mathbb{E}_{\mathcal{B}_L'}[Z_s] = \frac{1}{3^n} \left[ \sum_{\substack{ \vec{\gamma}^s_{ij}=II}} \text{wt}(\gamma) +  \frac{1}{2}(1+\alpha + \beta)\left(\sum_{\substack{ \vec{\gamma}^s_{ij}=IS}} \text{wt}(\gamma) + \sum_{\substack{ \vec{\gamma}^s_{ij}=SI}} \text{wt}(\gamma)\right)+(\alpha + \beta)\sum_{\substack{ \vec{\gamma}^s_{ij}=SS}} \text{wt}(\gamma)\right].
\end{equation}
Since $2/5 < 1/2$, we have $Z_{s+1}' > Z_{s+1}$.

\begin{align}
\begin{tikzcd}[transparent,row sep=0.7em]
& \gate[5]{C} & \qw & \qw & \qw & \qw \\
& \qw & \qw &\qw & \qw & \qw \\
& \qw& \gate[2]{U_2} & \gate{\mathcal{E}} & \gate{U_1} & \qw\\
&\qw & & \qw & \qw & \qw \\
& \qw & \qw &\qw &\qw & \qw
\end{tikzcd}
\to
\frac{1}{2}
\begin{tikzcd}[transparent,row sep=0.7em]
& \gate[5]{C} & \qw & \qw & \qw & \qw \\
& \qw & \qw &\qw & \qw & \qw \\
& \qw& \gate{U_1} & \gate{\mathcal{E}} & \gate{U_1} & \qw\\
&\qw & \gate{U_1} & \qw & \qw & \qw \\
& \qw & \qw &\qw &\qw & \qw
\end{tikzcd} 
+ \frac{1}{2}
\begin{tikzcd}[transparent,row sep=0.7em]
& \gate[5]{C} & \qw & \qw & \qw & \qw & \qw \\
& \qw & \qw & \qw &\qw & \qw & \qw \\
& \qw &\gate[swap]{} &\gate{U_1} & \gate{\mathcal{E}} & \gate{U_1} & \qw\\
&\qw & &\gate{U_1} & \qw & \qw & \qw \\
& \qw & \qw &\qw &\qw & \qw &\qw
\end{tikzcd}
\end{align}


\bigskip

\paragraph{Type C:} These gates have dephasing noise on both legs. The noise is followed by single-qubit gates for both legs. The combined channel has the form $M_{U_1}^{(j)}\circ \mathcal{E}^{(j)} \circ M_{U_1}^{(i)} \circ \mathcal{E}^{(i)} \circ M_{U_2}$. The bit-strings evolve as follows:
\begin{align}
    II \to II , \qquad SS \to \alpha^2 II + \alpha\beta IS + \alpha\beta SI + \beta^2 SS, \qquad IS \to \frac{2}{5} ( (1+\alpha^2) II + \alpha\beta IS + \alpha\beta SI + \beta^2 SS).
\end{align}
The average collision probability is thus given by
\begin{equation}
       \mathbb{E}_{\mathcal{B}_L}[Z_{s+1}] = \frac{1}{3^n} \left[ \sum_{\substack{ \vec{\gamma}^s_{ij}=II}} \text{wt}(\gamma) +  \frac{2}{5}\left(1+(\alpha + \beta)^2\right)\left(\sum_{\substack{ \vec{\gamma}^s_{ij}=IS}} \text{wt}(\gamma) + \sum_{\substack{ \vec{\gamma}^s_{ij}=SI}} \text{wt}(\gamma)\right)+(\alpha + \beta)^2 \sum_{\substack{ \vec{\gamma}^s_{ij}=SS}} \text{wt}(\gamma)\right].
\end{equation}

\begin{align}
\begin{tikzcd}[transparent,row sep=0.7em]
& \gate[5]{C} & \qw & \qw & \qw & \qw \\
& \qw & \qw &\qw & \qw & \qw \\
& \qw& \gate[2]{U_2} & \gate{\mathcal{E}} & \gate{U_1} & \qw\\
&\qw & & \gate{\mathcal{E}} & \gate{U_1} & \qw \\
& \qw & \qw &\qw &\qw & \qw
\end{tikzcd}
\to
\frac{1}{2}
\begin{tikzcd}[transparent,row sep=0.7em]
& \gate[5]{C} & \qw & \qw & \qw & \qw \\
& \qw & \qw &\qw & \qw & \qw \\
& \qw& \gate{U_1} & \gate{\mathcal{E}} & \gate{U_1} & \qw\\
&\qw & \gate{U_1} & \gate{\mathcal{E}} & \gate{U_1} & \qw \\
& \qw & \qw &\qw &\qw & \qw
\end{tikzcd} 
+ \frac{1}{2}
\begin{tikzcd}[transparent,row sep=0.7em]
& \gate[5]{C} & \qw & \qw & \qw & \qw & \qw \\
& \qw & \qw & \qw &\qw & \qw & \qw \\
& \qw &\gate[swap]{} &\gate{U_1} & \gate{\mathcal{E}} & \gate{U_1} & \qw\\
&\qw & &\gate{U_1} & \gate{\mathcal{E}} & \gate{U_1} & \qw \\
& \qw & \qw &\qw &\qw & \qw &\qw
\end{tikzcd} 
\end{align}

%

If instead we replace the two-qubit gate with two single-qubit Haar random gates preceded by a SWAP gate with probability 1/2 (as shown in the diagram above), we get channels of the form $M_{U_1}^{(j)}\circ \mathcal{E}^{(j)} \circ \mathcal{E}^{(i)} \circ M_{U_1}^{(j)} \circ   \mathcal{E}^{(i)} \circ M_{U_1}^{(i)} \circ \frac{1}{2}(SWAP+\mathbb{I})$, which, up to the SWAP gate, is same as the composite channel in Lemma \ref{lem:single-qubit} applied to both qubits. The states evolve as
\begin{equation}
    II \to II , \qquad SS \to \alpha^2 II + \alpha\beta IS + \alpha\beta SI + \beta^2 SS, \qquad IS \to \alpha II + \frac{\beta}{2} (IS+SI), \qquad SI \to \alpha II + \frac{\beta}{2} (SI+IS).
\end{equation}
The collision probability of the modified circuit is thus
\begin{equation}
       \mathbb{E}_{\mathcal{B}_L'}[Z_{s+1}] = \frac{1}{3^n} \left[ \sum_{\substack{ \vec{\gamma}^s_{ij}=II}} \text{wt}(\gamma) +  (\alpha + \beta)\left(\sum_{\substack{ \vec{\gamma}^s_{ij}=IS}} \text{wt}(\gamma) + \sum_{\substack{ \vec{\gamma}^s_{ij}=SI}} \text{wt}(\gamma)\right)+(\alpha + \beta)^2 \sum_{\substack{ \vec{\gamma}^s_{ij}=SS}} \text{wt}(\gamma)\right].
\end{equation}
Since $(\alpha + \beta) = 1+4q(1-q)/3$, $\alpha + \beta \in [1,4/3] \subset (\frac{1}{2}, 2)$, we have that $(2/5)\left(1+(\alpha + \beta)^2\right)< (\alpha + \beta)$, and therefore $Z_{s+1}' > Z_s$.

Starting from the input state, we can use the replacement procedure discussed above to iteratively define a  new circuit ensemble $\mathcal{B}_L'$ with the gates drawn at random independently of $L$, composed solely of single-qubit gates and SWAP gates, that has an equal or higher gate-averaged collision probability. This concludes the proof of the lemma.
\end{proof}

\section{Proof of \cref{lem_smallvariance}} \label{apx_variance}
We restate the lemma here for reference.
\setcounter{lemma}{2}
\begin{lemma}
The variance $\sigma^2 := \mathbb{E}_\mathcal{B}[x^2]-\mathbb{E}_\mathcal{B}[x]^2$ of the random variable $x := n\log 2 + \frac{1}{n!} \sum_{\sigma \in S_n} A_{\sigma} + \frac{1}{4 L_{\max}(d)}\sum_j  \langle Z_j \rangle^2$ in \cref{eq_lowerbound} satisfies $\sigma^2 \leq 2n$.
\end{lemma}

\begin{proof}[Proof of \cref{lem_smallvariance}]
The variance is
\begin{align}
\mathbb{E}_\mathcal{B}[x^2] - \mathbb{E}_\mathcal{B}[x]^2 = \frac{1}{n!^2} \mathbb{E}_\mathcal{B}\left[\left(\sum_\sigma A_\sigma\right)^2 \right] + \frac{1}{2L_\mathrm{max}(d)n!}\mathbb{E}_\mathcal{B}\left[\sum_\sigma A_\sigma \sum_j \expval{Z_j}^2\right] + \frac{1}{16 L_\mathrm{max}(d)^2}\mathbb{E}_\mathcal{B}\left[\sum_{i,j} \expval{Z_i}^2 \expval{Z_j}^2 \right] \nonumber \\
- \frac{1}{16L_\mathrm{max}(d)^2}\mathbb{E}_\mathcal{B}\left[\sum_i \expval{Z_i}^2\right]\mathbb{E}_\mathcal{B}\left[\sum_j \expval{Z_j}^2\right]. \label{eq_variance}
\end{align}

We will use the covariance bound
\begin{align}
\mathbb{E}[X Y] &= \sum_i \sqrt{p_i}x_i \sqrt{p_i} y_i
\\ &\leq \sqrt{\left(\sum_i p_i x_i^2\right)\left(\sum_j p_j y_j^2\right)}
\\ &= \sqrt{\mathbb{E}\left[X^2\right] \mathbb{E}\left[Y^2\right]}.
\end{align}
Applying this to the first term gives
\begin{align}
\frac{1}{n!^2} \mathbb{E}_\mathcal{B}\left[\left(\sum_\sigma A_\sigma\right)^2 \right] = \frac{1}{n!^2} \sum_{\sigma, \tau} \mathbb{E}_\mathcal{B}\left[A_\sigma A_\tau \right]
\\ \leq \frac{1}{n!^2} \sum_{\sigma,\tau} \sqrt{\mathbb{E}_\mathcal{B}[A_\sigma^2]\mathbb{E}_\mathcal{B}[A_\tau^2]}.
\end{align}
Now consider
\begin{align}
\mathbb{E}_\mathcal{B}\left[A_\sigma^2\right] &= \mathbb{E}_\mathcal{B}\left[\sum_{i,j} \expval{Z_{\sigma(i)}}_{\sigma(1)\ldots \sigma(i-1)} \expval{Z_{\sigma(j)}}_{\sigma(1)\ldots \sigma(j-1)} \right]
\\ &= \mathbb{E}_\mathcal{B}\left[\sum_{i} \expval{Z_{\sigma(i)}}_{\sigma(1)\ldots \sigma(i-1)}^2 \right] +  2 \mathbb{E}_\mathcal{B}\left[\sum_{i<j} \expval{Z_{\sigma(i)}}_{\sigma(1)\ldots \sigma(i-1)} \expval{Z_{\sigma(j)}}_{\sigma(1)\ldots \sigma(j-1)} \right]
\\ & \leq n
\end{align}
because the latter term is 0 (since $\sigma(j) \notin \{\sigma(1),\ldots \sigma(i)\}$ because $j>i$).
Therefore
\begin{align}
\frac{1}{n!^2} \mathbb{E}_\mathcal{B}\left[\left(\sum_\sigma A_\sigma\right)^2 \right] \leq \frac{1}{n!^2} \sum_{\sigma,\tau} \sqrt{n\cdot n} = n.
\end{align}

The next term is
\begin{align}
\frac{1}{2L_\mathrm{max}(d)n!}\mathbb{E}_\mathcal{B}\left[\sum_\sigma A_\sigma \sum_j \expval{Z_j}^2\right] = \frac{1}{2L_\mathrm{max}(d)n!}\mathbb{E}_\mathcal{B}\left[\sum_\sigma \sum_i \expval{Z_{\sigma(i)}}_{\sigma(1)\ldots \sigma(i-1)} \sum_j \expval{Z_j}^2\right].
\end{align}
The expectation is zero whenever $\sigma(i) \neq j$ or $j \notin \{\sigma(1),\ldots \sigma(i)\}$.
The latter condition is superseded by the first.
The only (potentially) nonzero term is when $\sigma(i)=j$.
But this term is zero as well since
\begin{align}
\mathbb{E}_\mathcal{B}[\expval{Z_j}_{J_j}\expval{Z_j}^2] = 0
\end{align}
because the 3-copy Haar-average is
\begin{align}
\int \mathrm{d}U
\begin{tikzcd}[row sep={0.8cm,between origins},transparent]
\qw & \gate{U} & \gate[3]{ZZZ} & \gate{U^\dagger} & \qw\\
\qw & \gate{U} & \qw & \gate{U^\dagger} & \qw\\
\qw & \gate{U} & \qw & \gate{U^\dagger} & \qw
\end{tikzcd}
= 0.
\end{align}
This brings us to the final term
\begin{align}
\frac{1}{16 L_\mathrm{max}(d)^2}\mathbb{E}_\mathcal{B}\left[\sum_{i,j} \expval{Z_i}^2 \expval{Z_j}^2 \right] - \frac{1}{16L_\mathrm{max}(d)^2}\mathbb{E}_\mathcal{B}\left[\sum_i \expval{Z_i}^2\right]\mathbb{E}\left[\sum_j \expval{Z_j}^2\right]
\\ = \frac{1}{16 L_\mathrm{max}(d)^2} \sum_{i,j} \left(\mathbb{E}_\mathcal{B}\left[\expval{Z_i}^2 \expval{Z_j}^2 \right] - \mathbb{E}_\mathcal{B}\left[\expval{Z_i}^2\right]\mathbb{E}_\mathcal{B}\left[\expval{Z_j}^2\right]\right).
\end{align}
For a fixed $i$, $j$, the term $\mathbb{E}_\mathcal{B}\left[\expval{Z_i}^2 \expval{Z_j}^2 \right] - \mathbb{E}_\mathcal{B}\left[\expval{Z_i}^2\right]\mathbb{E}_\mathcal{B}\left[\expval{Z_j}^2\right]$ is zero unless $j$ is in $L_d \circ L_d^\dag(i)$.
Otherwise, it is at most 1.
This means the sum is at most
\begin{align}
\frac{1}{16 L_\mathrm{max}(d)^2} \sum_{i} \abs{L_d \circ L_d^\dag(i)}
\\ \leq \frac{n}{16 L_\mathrm{max}(d)} \leq n.
\end{align}
The variance in \cref{eq_variance} is therefore at most $2n$.
\end{proof}

\begin{thebibliography}{57}%
\makeatletter
\providecommand \@ifxundefined [1]{%
 \@ifx{#1\undefined}
}%
\providecommand \@ifnum [1]{%
 \ifnum #1\expandafter \@firstoftwo
 \else \expandafter \@secondoftwo
 \fi
}%
\providecommand \@ifx [1]{%
 \ifx #1\expandafter \@firstoftwo
 \else \expandafter \@secondoftwo
 \fi
}%
\providecommand \natexlab [1]{#1}%
\providecommand \enquote  [1]{``#1''}%
\providecommand \bibnamefont  [1]{#1}%
\providecommand \bibfnamefont [1]{#1}%
\providecommand \citenamefont [1]{#1}%
\providecommand \href@noop [0]{\@secondoftwo}%
\providecommand \href [0]{\begingroup \@sanitize@url \@href}%
\providecommand \@href[1]{\@@startlink{#1}\@@href}%
\providecommand \@@href[1]{\endgroup#1\@@endlink}%
\providecommand \@sanitize@url [0]{\catcode `\\12\catcode `\$12\catcode
  `\&12\catcode `\#12\catcode `\^12\catcode `\_12\catcode `\%12\relax}%
\providecommand \@@startlink[1]{}%
\providecommand \@@endlink[0]{}%
\providecommand \url  [0]{\begingroup\@sanitize@url \@url }%
\providecommand \@url [1]{\endgroup\@href {#1}{\urlprefix }}%
\providecommand \urlprefix  [0]{URL }%
\providecommand \Eprint [0]{\href }%
\providecommand \doibase [0]{https://doi.org/}%
\providecommand \selectlanguage [0]{\@gobble}%
\providecommand \bibinfo  [0]{\@secondoftwo}%
\providecommand \bibfield  [0]{\@secondoftwo}%
\providecommand \translation [1]{[#1]}%
\providecommand \BibitemOpen [0]{}%
\providecommand \bibitemStop [0]{}%
\providecommand \bibitemNoStop [0]{.\EOS\space}%
\providecommand \EOS [0]{\spacefactor3000\relax}%
\providecommand \BibitemShut  [1]{\csname bibitem#1\endcsname}%
\let\auto@bib@innerbib\@empty
\bibitem [{\citenamefont {Preskill}(2018)}]{Preskill2018}%
  \BibitemOpen
  \bibfield  {author} {\bibinfo {author} {\bibfnamefont {J.}~\bibnamefont
  {Preskill}},\ }\bibfield  {title} {\bibinfo {title} {Quantum {{Computing}} in
  the {{NISQ}} era and beyond},\ }\href
  {https://doi.org/10.22331/q-2018-08-06-79} {\bibfield  {journal} {\bibinfo
  {journal} {Quantum}\ }\textbf {\bibinfo {volume} {2}},\ \bibinfo {pages} {79}
  (\bibinfo {year} {2018})}\BibitemShut {NoStop}%
\bibitem [{\citenamefont {Arute~et al.}(2019)}]{Arute2019}%
  \BibitemOpen
  \bibfield  {author} {\bibinfo {author} {\bibfnamefont {F.}~\bibnamefont
  {Arute~et al.}},\ }\bibfield  {title} {\bibinfo {title} {Quantum supremacy
  using a programmable superconducting processor},\ }\href
  {https://doi.org/10.1038/s41586-019-1666-5} {\bibfield  {journal} {\bibinfo
  {journal} {Nature}\ }\textbf {\bibinfo {volume} {574}},\ \bibinfo {pages}
  {505} (\bibinfo {year} {2019})}\BibitemShut {NoStop}%
\bibitem [{\citenamefont {Zhong}\ \emph {et~al.}(2020)\citenamefont {Zhong},
  \citenamefont {Wang}, \citenamefont {Deng}, \citenamefont {Chen},
  \citenamefont {Peng}, \citenamefont {Luo}, \citenamefont {Qin}, \citenamefont
  {Wu}, \citenamefont {Ding}, \citenamefont {Hu}, \citenamefont {Hu},
  \citenamefont {Yang}, \citenamefont {Zhang}, \citenamefont {Li},
  \citenamefont {Li}, \citenamefont {Jiang}, \citenamefont {Gan}, \citenamefont
  {Yang}, \citenamefont {You}, \citenamefont {Wang}, \citenamefont {Li},
  \citenamefont {Liu}, \citenamefont {Lu},\ and\ \citenamefont
  {Pan}}]{Zhong2020}%
  \BibitemOpen
  \bibfield  {author} {\bibinfo {author} {\bibfnamefont {H.-S.}\ \bibnamefont
  {Zhong}}, \bibinfo {author} {\bibfnamefont {H.}~\bibnamefont {Wang}},
  \bibinfo {author} {\bibfnamefont {Y.-H.}\ \bibnamefont {Deng}}, \bibinfo
  {author} {\bibfnamefont {M.-C.}\ \bibnamefont {Chen}}, \bibinfo {author}
  {\bibfnamefont {L.-C.}\ \bibnamefont {Peng}}, \bibinfo {author}
  {\bibfnamefont {Y.-H.}\ \bibnamefont {Luo}}, \bibinfo {author} {\bibfnamefont
  {J.}~\bibnamefont {Qin}}, \bibinfo {author} {\bibfnamefont {D.}~\bibnamefont
  {Wu}}, \bibinfo {author} {\bibfnamefont {X.}~\bibnamefont {Ding}}, \bibinfo
  {author} {\bibfnamefont {Y.}~\bibnamefont {Hu}}, \bibinfo {author}
  {\bibfnamefont {P.}~\bibnamefont {Hu}}, \bibinfo {author} {\bibfnamefont
  {X.-Y.}\ \bibnamefont {Yang}}, \bibinfo {author} {\bibfnamefont {W.-J.}\
  \bibnamefont {Zhang}}, \bibinfo {author} {\bibfnamefont {H.}~\bibnamefont
  {Li}}, \bibinfo {author} {\bibfnamefont {Y.}~\bibnamefont {Li}}, \bibinfo
  {author} {\bibfnamefont {X.}~\bibnamefont {Jiang}}, \bibinfo {author}
  {\bibfnamefont {L.}~\bibnamefont {Gan}}, \bibinfo {author} {\bibfnamefont
  {G.}~\bibnamefont {Yang}}, \bibinfo {author} {\bibfnamefont {L.}~\bibnamefont
  {You}}, \bibinfo {author} {\bibfnamefont {Z.}~\bibnamefont {Wang}}, \bibinfo
  {author} {\bibfnamefont {L.}~\bibnamefont {Li}}, \bibinfo {author}
  {\bibfnamefont {N.-L.}\ \bibnamefont {Liu}}, \bibinfo {author} {\bibfnamefont
  {C.-Y.}\ \bibnamefont {Lu}},\ and\ \bibinfo {author} {\bibfnamefont {J.-W.}\
  \bibnamefont {Pan}},\ }\bibfield  {title} {\bibinfo {title} {Quantum
  computational advantage using photons},\ }\href
  {https://doi.org/10.1126/science.abe8770} {\bibfield  {journal} {\bibinfo
  {journal} {Science}\ }\textbf {\bibinfo {volume} {370}},\ \bibinfo {pages}
  {1460} (\bibinfo {year} {2020})}\BibitemShut {NoStop}%
\bibitem [{\citenamefont {Zhong}\ \emph {et~al.}(2021)\citenamefont {Zhong},
  \citenamefont {Deng}, \citenamefont {Qin}, \citenamefont {Wang},
  \citenamefont {Chen}, \citenamefont {Peng}, \citenamefont {Luo},
  \citenamefont {Wu}, \citenamefont {Gong}, \citenamefont {Su}, \citenamefont
  {Hu}, \citenamefont {Hu}, \citenamefont {Yang}, \citenamefont {Zhang},
  \citenamefont {Li}, \citenamefont {Li}, \citenamefont {Jiang}, \citenamefont
  {Gan}, \citenamefont {Yang}, \citenamefont {You}, \citenamefont {Wang},
  \citenamefont {Li}, \citenamefont {Liu}, \citenamefont {Renema},
  \citenamefont {Lu},\ and\ \citenamefont {Pan}}]{Zhong2021}%
  \BibitemOpen
  \bibfield  {author} {\bibinfo {author} {\bibfnamefont {H.-S.}\ \bibnamefont
  {Zhong}}, \bibinfo {author} {\bibfnamefont {Y.-H.}\ \bibnamefont {Deng}},
  \bibinfo {author} {\bibfnamefont {J.}~\bibnamefont {Qin}}, \bibinfo {author}
  {\bibfnamefont {H.}~\bibnamefont {Wang}}, \bibinfo {author} {\bibfnamefont
  {M.-C.}\ \bibnamefont {Chen}}, \bibinfo {author} {\bibfnamefont {L.-C.}\
  \bibnamefont {Peng}}, \bibinfo {author} {\bibfnamefont {Y.-H.}\ \bibnamefont
  {Luo}}, \bibinfo {author} {\bibfnamefont {D.}~\bibnamefont {Wu}}, \bibinfo
  {author} {\bibfnamefont {S.-Q.}\ \bibnamefont {Gong}}, \bibinfo {author}
  {\bibfnamefont {H.}~\bibnamefont {Su}}, \bibinfo {author} {\bibfnamefont
  {Y.}~\bibnamefont {Hu}}, \bibinfo {author} {\bibfnamefont {P.}~\bibnamefont
  {Hu}}, \bibinfo {author} {\bibfnamefont {X.-Y.}\ \bibnamefont {Yang}},
  \bibinfo {author} {\bibfnamefont {W.-J.}\ \bibnamefont {Zhang}}, \bibinfo
  {author} {\bibfnamefont {H.}~\bibnamefont {Li}}, \bibinfo {author}
  {\bibfnamefont {Y.}~\bibnamefont {Li}}, \bibinfo {author} {\bibfnamefont
  {X.}~\bibnamefont {Jiang}}, \bibinfo {author} {\bibfnamefont
  {L.}~\bibnamefont {Gan}}, \bibinfo {author} {\bibfnamefont {G.}~\bibnamefont
  {Yang}}, \bibinfo {author} {\bibfnamefont {L.}~\bibnamefont {You}}, \bibinfo
  {author} {\bibfnamefont {Z.}~\bibnamefont {Wang}}, \bibinfo {author}
  {\bibfnamefont {L.}~\bibnamefont {Li}}, \bibinfo {author} {\bibfnamefont
  {N.-L.}\ \bibnamefont {Liu}}, \bibinfo {author} {\bibfnamefont {J.~J.}\
  \bibnamefont {Renema}}, \bibinfo {author} {\bibfnamefont {C.-Y.}\
  \bibnamefont {Lu}},\ and\ \bibinfo {author} {\bibfnamefont {J.-W.}\
  \bibnamefont {Pan}},\ }\bibfield  {title} {\bibinfo {title}
  {Phase-{{Programmable Gaussian Boson Sampling Using Stimulated Squeezed
  Light}}},\ }\href {https://doi.org/10.1103/PhysRevLett.127.180502} {\bibfield
   {journal} {\bibinfo  {journal} {Phys. Rev. Lett.}\ }\textbf {\bibinfo
  {volume} {127}},\ \bibinfo {pages} {180502} (\bibinfo {year}
  {2021})}\BibitemShut {NoStop}%
\bibitem [{\citenamefont {Wu}\ \emph {et~al.}(2021)\citenamefont {Wu},
  \citenamefont {Bao}, \citenamefont {Cao}, \citenamefont {Chen}, \citenamefont
  {Chen}, \citenamefont {Chen}, \citenamefont {Chung}, \citenamefont {Deng},
  \citenamefont {Du}, \citenamefont {Fan}, \citenamefont {Gong}, \citenamefont
  {Guo}, \citenamefont {Guo}, \citenamefont {Guo}, \citenamefont {Han},
  \citenamefont {Hong}, \citenamefont {Huang}, \citenamefont {Huo},
  \citenamefont {Li}, \citenamefont {Li}, \citenamefont {Li}, \citenamefont
  {Li}, \citenamefont {Liang}, \citenamefont {Lin}, \citenamefont {Lin},
  \citenamefont {Qian}, \citenamefont {Qiao}, \citenamefont {Rong},
  \citenamefont {Su}, \citenamefont {Sun}, \citenamefont {Wang}, \citenamefont
  {Wang}, \citenamefont {Wu}, \citenamefont {Xu}, \citenamefont {Yan},
  \citenamefont {Yang}, \citenamefont {Yang}, \citenamefont {Ye}, \citenamefont
  {Yin}, \citenamefont {Ying}, \citenamefont {Yu}, \citenamefont {Zha},
  \citenamefont {Zhang}, \citenamefont {Zhang}, \citenamefont {Zhang},
  \citenamefont {Zhang}, \citenamefont {Zhao}, \citenamefont {Zhao},
  \citenamefont {Zhou}, \citenamefont {Zhu}, \citenamefont {Lu}, \citenamefont
  {Peng}, \citenamefont {Zhu},\ and\ \citenamefont {Pan}}]{Wu2021}%
  \BibitemOpen
  \bibfield  {author} {\bibinfo {author} {\bibfnamefont {Y.}~\bibnamefont
  {Wu}}, \bibinfo {author} {\bibfnamefont {W.-S.}\ \bibnamefont {Bao}},
  \bibinfo {author} {\bibfnamefont {S.}~\bibnamefont {Cao}}, \bibinfo {author}
  {\bibfnamefont {F.}~\bibnamefont {Chen}}, \bibinfo {author} {\bibfnamefont
  {M.-C.}\ \bibnamefont {Chen}}, \bibinfo {author} {\bibfnamefont
  {X.}~\bibnamefont {Chen}}, \bibinfo {author} {\bibfnamefont {T.-H.}\
  \bibnamefont {Chung}}, \bibinfo {author} {\bibfnamefont {H.}~\bibnamefont
  {Deng}}, \bibinfo {author} {\bibfnamefont {Y.}~\bibnamefont {Du}}, \bibinfo
  {author} {\bibfnamefont {D.}~\bibnamefont {Fan}}, \bibinfo {author}
  {\bibfnamefont {M.}~\bibnamefont {Gong}}, \bibinfo {author} {\bibfnamefont
  {C.}~\bibnamefont {Guo}}, \bibinfo {author} {\bibfnamefont {C.}~\bibnamefont
  {Guo}}, \bibinfo {author} {\bibfnamefont {S.}~\bibnamefont {Guo}}, \bibinfo
  {author} {\bibfnamefont {L.}~\bibnamefont {Han}}, \bibinfo {author}
  {\bibfnamefont {L.}~\bibnamefont {Hong}}, \bibinfo {author} {\bibfnamefont
  {H.-L.}\ \bibnamefont {Huang}}, \bibinfo {author} {\bibfnamefont {Y.-H.}\
  \bibnamefont {Huo}}, \bibinfo {author} {\bibfnamefont {L.}~\bibnamefont
  {Li}}, \bibinfo {author} {\bibfnamefont {N.}~\bibnamefont {Li}}, \bibinfo
  {author} {\bibfnamefont {S.}~\bibnamefont {Li}}, \bibinfo {author}
  {\bibfnamefont {Y.}~\bibnamefont {Li}}, \bibinfo {author} {\bibfnamefont
  {F.}~\bibnamefont {Liang}}, \bibinfo {author} {\bibfnamefont
  {C.}~\bibnamefont {Lin}}, \bibinfo {author} {\bibfnamefont {J.}~\bibnamefont
  {Lin}}, \bibinfo {author} {\bibfnamefont {H.}~\bibnamefont {Qian}}, \bibinfo
  {author} {\bibfnamefont {D.}~\bibnamefont {Qiao}}, \bibinfo {author}
  {\bibfnamefont {H.}~\bibnamefont {Rong}}, \bibinfo {author} {\bibfnamefont
  {H.}~\bibnamefont {Su}}, \bibinfo {author} {\bibfnamefont {L.}~\bibnamefont
  {Sun}}, \bibinfo {author} {\bibfnamefont {L.}~\bibnamefont {Wang}}, \bibinfo
  {author} {\bibfnamefont {S.}~\bibnamefont {Wang}}, \bibinfo {author}
  {\bibfnamefont {D.}~\bibnamefont {Wu}}, \bibinfo {author} {\bibfnamefont
  {Y.}~\bibnamefont {Xu}}, \bibinfo {author} {\bibfnamefont {K.}~\bibnamefont
  {Yan}}, \bibinfo {author} {\bibfnamefont {W.}~\bibnamefont {Yang}}, \bibinfo
  {author} {\bibfnamefont {Y.}~\bibnamefont {Yang}}, \bibinfo {author}
  {\bibfnamefont {Y.}~\bibnamefont {Ye}}, \bibinfo {author} {\bibfnamefont
  {J.}~\bibnamefont {Yin}}, \bibinfo {author} {\bibfnamefont {C.}~\bibnamefont
  {Ying}}, \bibinfo {author} {\bibfnamefont {J.}~\bibnamefont {Yu}}, \bibinfo
  {author} {\bibfnamefont {C.}~\bibnamefont {Zha}}, \bibinfo {author}
  {\bibfnamefont {C.}~\bibnamefont {Zhang}}, \bibinfo {author} {\bibfnamefont
  {H.}~\bibnamefont {Zhang}}, \bibinfo {author} {\bibfnamefont
  {K.}~\bibnamefont {Zhang}}, \bibinfo {author} {\bibfnamefont
  {Y.}~\bibnamefont {Zhang}}, \bibinfo {author} {\bibfnamefont
  {H.}~\bibnamefont {Zhao}}, \bibinfo {author} {\bibfnamefont {Y.}~\bibnamefont
  {Zhao}}, \bibinfo {author} {\bibfnamefont {L.}~\bibnamefont {Zhou}}, \bibinfo
  {author} {\bibfnamefont {Q.}~\bibnamefont {Zhu}}, \bibinfo {author}
  {\bibfnamefont {C.-Y.}\ \bibnamefont {Lu}}, \bibinfo {author} {\bibfnamefont
  {C.-Z.}\ \bibnamefont {Peng}}, \bibinfo {author} {\bibfnamefont
  {X.}~\bibnamefont {Zhu}},\ and\ \bibinfo {author} {\bibfnamefont {J.-W.}\
  \bibnamefont {Pan}},\ }\bibfield  {title} {\bibinfo {title} {Strong quantum
  computational advantage using a superconducting quantum processor},\ }\href
  {https://doi.org/10.1103/PhysRevLett.127.180501} {\bibfield  {journal}
  {\bibinfo  {journal} {Phys. Rev. Lett.}\ }\textbf {\bibinfo {volume} {127}},\
  \bibinfo {pages} {180501} (\bibinfo {year} {2021})}\BibitemShut {NoStop}%
\bibitem [{\citenamefont {Brown}\ and\ \citenamefont
  {Fawzi}(2012)}]{Brown2012}%
  \BibitemOpen
  \bibfield  {author} {\bibinfo {author} {\bibfnamefont {W.}~\bibnamefont
  {Brown}}\ and\ \bibinfo {author} {\bibfnamefont {O.}~\bibnamefont {Fawzi}},\
  }\bibfield  {title} {\bibinfo {title} {Scrambling speed of random quantum
  circuits},\ }\href {http://arxiv.org/abs/1210.6644} {\  (\bibinfo {year}
  {2012})},\ \Eprint {https://arxiv.org/abs/1210.6644} {arXiv:1210.6644
  [hep-th, physics:quant-ph]} \BibitemShut {NoStop}%
\bibitem [{\citenamefont {Brown}\ and\ \citenamefont
  {Fawzi}(2015)}]{Brown2015}%
  \BibitemOpen
  \bibfield  {author} {\bibinfo {author} {\bibfnamefont {W.}~\bibnamefont
  {Brown}}\ and\ \bibinfo {author} {\bibfnamefont {O.}~\bibnamefont {Fawzi}},\
  }\bibfield  {title} {\bibinfo {title} {Decoupling with random quantum
  circuits},\ }\href {https://doi.org/10.1007/s00220-015-2470-1} {\bibfield
  {journal} {\bibinfo  {journal} {Commun. Math. Phys.}\ }\textbf {\bibinfo
  {volume} {340}},\ \bibinfo {pages} {867} (\bibinfo {year}
  {2015})}\BibitemShut {NoStop}%
\bibitem [{\citenamefont {Brand{\~a}o}\ \emph {et~al.}(2021)\citenamefont
  {Brand{\~a}o}, \citenamefont {Chemissany}, \citenamefont {{Hunter-Jones}},
  \citenamefont {Kueng},\ and\ \citenamefont {Preskill}}]{Brandao2021}%
  \BibitemOpen
  \bibfield  {author} {\bibinfo {author} {\bibfnamefont {F.~G. S.~L.}\
  \bibnamefont {Brand{\~a}o}}, \bibinfo {author} {\bibfnamefont
  {W.}~\bibnamefont {Chemissany}}, \bibinfo {author} {\bibfnamefont
  {N.}~\bibnamefont {{Hunter-Jones}}}, \bibinfo {author} {\bibfnamefont
  {R.}~\bibnamefont {Kueng}},\ and\ \bibinfo {author} {\bibfnamefont
  {J.}~\bibnamefont {Preskill}},\ }\bibfield  {title} {\bibinfo {title} {Models
  of {{Quantum Complexity Growth}}},\ }\href
  {https://doi.org/10.1103/PRXQuantum.2.030316} {\bibfield  {journal} {\bibinfo
   {journal} {PRX Quantum}\ }\textbf {\bibinfo {volume} {2}},\ \bibinfo {pages}
  {030316} (\bibinfo {year} {2021})}\BibitemShut {NoStop}%
\bibitem [{\citenamefont {Haferkamp}\ \emph {et~al.}(2021)\citenamefont
  {Haferkamp}, \citenamefont {Faist}, \citenamefont {Kothakonda}, \citenamefont
  {Eisert},\ and\ \citenamefont {Halpern}}]{Haferkamp2021a}%
  \BibitemOpen
  \bibfield  {author} {\bibinfo {author} {\bibfnamefont {J.}~\bibnamefont
  {Haferkamp}}, \bibinfo {author} {\bibfnamefont {P.}~\bibnamefont {Faist}},
  \bibinfo {author} {\bibfnamefont {N.~B.~T.}\ \bibnamefont {Kothakonda}},
  \bibinfo {author} {\bibfnamefont {J.}~\bibnamefont {Eisert}},\ and\ \bibinfo
  {author} {\bibfnamefont {N.~Y.}\ \bibnamefont {Halpern}},\ }\bibfield
  {title} {\bibinfo {title} {Linear growth of quantum circuit complexity},\
  }\href {http://arxiv.org/abs/2106.05305} {\  (\bibinfo {year} {2021})},\
  \Eprint {https://arxiv.org/abs/2106.05305} {arXiv:2106.05305 [hep-th,
  physics:math-ph, physics:quant-ph]} \BibitemShut {NoStop}%
\bibitem [{\citenamefont {Potter}\ and\ \citenamefont
  {Vasseur}(2021)}]{Potter2021}%
  \BibitemOpen
  \bibfield  {author} {\bibinfo {author} {\bibfnamefont {A.~C.}\ \bibnamefont
  {Potter}}\ and\ \bibinfo {author} {\bibfnamefont {R.}~\bibnamefont
  {Vasseur}},\ }\bibfield  {title} {\bibinfo {title} {Entanglement dynamics in
  hybrid quantum circuits},\ }\href {http://arxiv.org/abs/2111.08018} {\
  (\bibinfo {year} {2021})},\ \Eprint {https://arxiv.org/abs/2111.08018}
  {arXiv:2111.08018 [cond-mat, physics:quant-ph]} \BibitemShut {NoStop}%
\bibitem [{\citenamefont {McClean}\ \emph {et~al.}(2018)\citenamefont
  {McClean}, \citenamefont {Boixo}, \citenamefont {Smelyanskiy}, \citenamefont
  {Babbush},\ and\ \citenamefont {Neven}}]{McClean2018}%
  \BibitemOpen
  \bibfield  {author} {\bibinfo {author} {\bibfnamefont {J.~R.}\ \bibnamefont
  {McClean}}, \bibinfo {author} {\bibfnamefont {S.}~\bibnamefont {Boixo}},
  \bibinfo {author} {\bibfnamefont {V.~N.}\ \bibnamefont {Smelyanskiy}},
  \bibinfo {author} {\bibfnamefont {R.}~\bibnamefont {Babbush}},\ and\ \bibinfo
  {author} {\bibfnamefont {H.}~\bibnamefont {Neven}},\ }\bibfield  {title}
  {\bibinfo {title} {Barren plateaus in quantum neural network training
  landscapes},\ }\href {https://doi.org/10.1038/s41467-018-07090-4} {\bibfield
  {journal} {\bibinfo  {journal} {Nat Commun}\ }\textbf {\bibinfo {volume}
  {9}},\ \bibinfo {pages} {4812} (\bibinfo {year} {2018})}\BibitemShut
  {NoStop}%
\bibitem [{\citenamefont {Wang}\ \emph
  {et~al.}(2021{\natexlab{a}})\citenamefont {Wang}, \citenamefont {Fontana},
  \citenamefont {Cerezo}, \citenamefont {Sharma}, \citenamefont {Sone},
  \citenamefont {Cincio},\ and\ \citenamefont {Coles}}]{Wang2021}%
  \BibitemOpen
  \bibfield  {author} {\bibinfo {author} {\bibfnamefont {S.}~\bibnamefont
  {Wang}}, \bibinfo {author} {\bibfnamefont {E.}~\bibnamefont {Fontana}},
  \bibinfo {author} {\bibfnamefont {M.}~\bibnamefont {Cerezo}}, \bibinfo
  {author} {\bibfnamefont {K.}~\bibnamefont {Sharma}}, \bibinfo {author}
  {\bibfnamefont {A.}~\bibnamefont {Sone}}, \bibinfo {author} {\bibfnamefont
  {L.}~\bibnamefont {Cincio}},\ and\ \bibinfo {author} {\bibfnamefont {P.~J.}\
  \bibnamefont {Coles}},\ }\bibfield  {title} {\bibinfo {title}
  {Noise-{{Induced Barren Plateaus}} in {{Variational Quantum Algorithms}}},\
  }\href {http://arxiv.org/abs/2007.14384} {\  (\bibinfo {year}
  {2021}{\natexlab{a}})},\ \Eprint {https://arxiv.org/abs/2007.14384}
  {arXiv:2007.14384 [quant-ph]} \BibitemShut {NoStop}%
\bibitem [{\citenamefont {Gao}\ and\ \citenamefont {Duan}(2018)}]{Gao2018}%
  \BibitemOpen
  \bibfield  {author} {\bibinfo {author} {\bibfnamefont {X.}~\bibnamefont
  {Gao}}\ and\ \bibinfo {author} {\bibfnamefont {L.}~\bibnamefont {Duan}},\
  }\bibfield  {title} {\bibinfo {title} {Efficient classical simulation of
  noisy quantum computation},\ }\href {http://arxiv.org/abs/1810.03176} {\
  (\bibinfo {year} {2018})},\ \Eprint {https://arxiv.org/abs/1810.03176}
  {arXiv:1810.03176 [quant-ph]} \BibitemShut {NoStop}%
\bibitem [{\citenamefont {Bouland}\ \emph {et~al.}(2021)\citenamefont
  {Bouland}, \citenamefont {Fefferman}, \citenamefont {Landau},\ and\
  \citenamefont {Liu}}]{Bouland2021}%
  \BibitemOpen
  \bibfield  {author} {\bibinfo {author} {\bibfnamefont {A.}~\bibnamefont
  {Bouland}}, \bibinfo {author} {\bibfnamefont {B.}~\bibnamefont {Fefferman}},
  \bibinfo {author} {\bibfnamefont {Z.}~\bibnamefont {Landau}},\ and\ \bibinfo
  {author} {\bibfnamefont {Y.}~\bibnamefont {Liu}},\ }\bibfield  {title}
  {\bibinfo {title} {Noise and the frontier of quantum supremacy},\ }\href
  {http://arxiv.org/abs/2102.01738} {\  (\bibinfo {year} {2021})},\ \Eprint
  {https://arxiv.org/abs/2102.01738} {arXiv:2102.01738 [quant-ph]} \BibitemShut
  {NoStop}%
\bibitem [{\citenamefont {Aharonov}\ \emph {et~al.}(1996)\citenamefont
  {Aharonov}, \citenamefont {{Ben-Or}}, \citenamefont {Impagliazzo},\ and\
  \citenamefont {Nisan}}]{Aharonov1996}%
  \BibitemOpen
  \bibfield  {author} {\bibinfo {author} {\bibfnamefont {D.}~\bibnamefont
  {Aharonov}}, \bibinfo {author} {\bibfnamefont {M.}~\bibnamefont {{Ben-Or}}},
  \bibinfo {author} {\bibfnamefont {R.}~\bibnamefont {Impagliazzo}},\ and\
  \bibinfo {author} {\bibfnamefont {N.}~\bibnamefont {Nisan}},\ }\bibfield
  {title} {\bibinfo {title} {Limitations of {{Noisy Reversible Computation}}},\
  }\href {http://arxiv.org/abs/quant-ph/9611028} {\  (\bibinfo {year}
  {1996})},\ \Eprint {https://arxiv.org/abs/quant-ph/9611028}
  {arXiv:quant-ph/9611028} \BibitemShut {NoStop}%
\bibitem [{\citenamefont {Aharonov}(2000)}]{Aharonov2000}%
  \BibitemOpen
  \bibfield  {author} {\bibinfo {author} {\bibfnamefont {D.}~\bibnamefont
  {Aharonov}},\ }\bibfield  {title} {\bibinfo {title} {Quantum to classical
  phase transition in noisy quantum computers},\ }\href
  {https://doi.org/10.1103/PhysRevA.62.062311} {\bibfield  {journal} {\bibinfo
  {journal} {Phys. Rev. A}\ }\textbf {\bibinfo {volume} {62}},\ \bibinfo
  {pages} {062311} (\bibinfo {year} {2000})}\BibitemShut {NoStop}%
\bibitem [{\citenamefont {Emerson}\ \emph {et~al.}(2005)\citenamefont
  {Emerson}, \citenamefont {Alicki},\ and\ \citenamefont
  {Zyczkowski}}]{Emerson2005}%
  \BibitemOpen
  \bibfield  {author} {\bibinfo {author} {\bibfnamefont {J.}~\bibnamefont
  {Emerson}}, \bibinfo {author} {\bibfnamefont {R.}~\bibnamefont {Alicki}},\
  and\ \bibinfo {author} {\bibfnamefont {K.}~\bibnamefont {Zyczkowski}},\
  }\bibfield  {title} {\bibinfo {title} {Scalable noise estimation with random
  unitary operators},\ }\href {https://doi.org/10.1088/1464-4266/7/10/021}
  {\bibfield  {journal} {\bibinfo  {journal} {J. Opt. B: Quantum Semiclass.
  Opt.}\ }\textbf {\bibinfo {volume} {7}},\ \bibinfo {pages} {S347} (\bibinfo
  {year} {2005})}\BibitemShut {NoStop}%
\bibitem [{\citenamefont {Boixo}\ \emph {et~al.}(2018)\citenamefont {Boixo},
  \citenamefont {Isakov}, \citenamefont {Smelyanskiy}, \citenamefont {Babbush},
  \citenamefont {Ding}, \citenamefont {Jiang}, \citenamefont {Bremner},
  \citenamefont {Martinis},\ and\ \citenamefont {Neven}}]{Boixo2018}%
  \BibitemOpen
  \bibfield  {author} {\bibinfo {author} {\bibfnamefont {S.}~\bibnamefont
  {Boixo}}, \bibinfo {author} {\bibfnamefont {S.~V.}\ \bibnamefont {Isakov}},
  \bibinfo {author} {\bibfnamefont {V.~N.}\ \bibnamefont {Smelyanskiy}},
  \bibinfo {author} {\bibfnamefont {R.}~\bibnamefont {Babbush}}, \bibinfo
  {author} {\bibfnamefont {N.}~\bibnamefont {Ding}}, \bibinfo {author}
  {\bibfnamefont {Z.}~\bibnamefont {Jiang}}, \bibinfo {author} {\bibfnamefont
  {M.~J.}\ \bibnamefont {Bremner}}, \bibinfo {author} {\bibfnamefont {J.~M.}\
  \bibnamefont {Martinis}},\ and\ \bibinfo {author} {\bibfnamefont
  {H.}~\bibnamefont {Neven}},\ }\bibfield  {title} {\bibinfo {title}
  {Characterizing quantum supremacy in near-term devices},\ }\href
  {https://doi.org/10.1038/s41567-018-0124-x} {\bibfield  {journal} {\bibinfo
  {journal} {Nat. Phys.}\ }\textbf {\bibinfo {volume} {14}},\ \bibinfo {pages}
  {595} (\bibinfo {year} {2018})}\BibitemShut {NoStop}%
\bibitem [{\citenamefont {Boixo}\ \emph {et~al.}(2017)\citenamefont {Boixo},
  \citenamefont {Smelyanskiy},\ and\ \citenamefont {Neven}}]{Boixo2017}%
  \BibitemOpen
  \bibfield  {author} {\bibinfo {author} {\bibfnamefont {S.}~\bibnamefont
  {Boixo}}, \bibinfo {author} {\bibfnamefont {V.~N.}\ \bibnamefont
  {Smelyanskiy}},\ and\ \bibinfo {author} {\bibfnamefont {H.}~\bibnamefont
  {Neven}},\ }\bibfield  {title} {\bibinfo {title} {Fourier analysis of
  sampling from noisy chaotic quantum circuits},\ }\href
  {http://arxiv.org/abs/1708.01875} {\  (\bibinfo {year} {2017})},\ \Eprint
  {https://arxiv.org/abs/1708.01875} {arXiv:1708.01875} \BibitemShut {NoStop}%
\bibitem [{\citenamefont {Liu}\ \emph {et~al.}(2021)\citenamefont {Liu},
  \citenamefont {Otten}, \citenamefont {Bassirianjahromi}, \citenamefont
  {Jiang},\ and\ \citenamefont {Fefferman}}]{Liu2021}%
  \BibitemOpen
  \bibfield  {author} {\bibinfo {author} {\bibfnamefont {Y.}~\bibnamefont
  {Liu}}, \bibinfo {author} {\bibfnamefont {M.}~\bibnamefont {Otten}}, \bibinfo
  {author} {\bibfnamefont {R.}~\bibnamefont {Bassirianjahromi}}, \bibinfo
  {author} {\bibfnamefont {L.}~\bibnamefont {Jiang}},\ and\ \bibinfo {author}
  {\bibfnamefont {B.}~\bibnamefont {Fefferman}},\ }\bibfield  {title} {\bibinfo
  {title} {Benchmarking near-term quantum computers via random circuit
  sampling},\ }\href {http://arxiv.org/abs/2105.05232} {\  (\bibinfo {year}
  {2021})},\ \Eprint {https://arxiv.org/abs/2105.05232} {arXiv:2105.05232
  [quant-ph]} \BibitemShut {NoStop}%
\bibitem [{\citenamefont {Stilck~Fran{\c c}a}\ and\ \citenamefont
  {{Garc{\'i}a-Patr{\'o}n}}(2021)}]{StilckFranca2021}%
  \BibitemOpen
  \bibfield  {author} {\bibinfo {author} {\bibfnamefont {D.}~\bibnamefont
  {Stilck~Fran{\c c}a}}\ and\ \bibinfo {author} {\bibfnamefont
  {R.}~\bibnamefont {{Garc{\'i}a-Patr{\'o}n}}},\ }\bibfield  {title} {\bibinfo
  {title} {Limitations of optimization algorithms on noisy quantum devices},\
  }\bibfield  {journal} {\bibinfo  {journal} {Nat. Phys.}\ }\href
  {https://doi.org/10.1038/s41567-021-01356-3} {10.1038/s41567-021-01356-3}
  (\bibinfo {year} {2021})\BibitemShut {NoStop}%
\bibitem [{\citenamefont {Wang}\ \emph
  {et~al.}(2021{\natexlab{b}})\citenamefont {Wang}, \citenamefont {Czarnik},
  \citenamefont {Arrasmith}, \citenamefont {Cerezo}, \citenamefont {Cincio},\
  and\ \citenamefont {Coles}}]{Wang2021a}%
  \BibitemOpen
  \bibfield  {author} {\bibinfo {author} {\bibfnamefont {S.}~\bibnamefont
  {Wang}}, \bibinfo {author} {\bibfnamefont {P.}~\bibnamefont {Czarnik}},
  \bibinfo {author} {\bibfnamefont {A.}~\bibnamefont {Arrasmith}}, \bibinfo
  {author} {\bibfnamefont {M.}~\bibnamefont {Cerezo}}, \bibinfo {author}
  {\bibfnamefont {L.}~\bibnamefont {Cincio}},\ and\ \bibinfo {author}
  {\bibfnamefont {P.~J.}\ \bibnamefont {Coles}},\ }\bibfield  {title} {\bibinfo
  {title} {Can {{Error Mitigation Improve Trainability}} of {{Noisy Variational
  Quantum Algorithms}}?},\ }\href {http://arxiv.org/abs/2109.01051} {\
  (\bibinfo {year} {2021}{\natexlab{b}})},\ \Eprint
  {https://arxiv.org/abs/2109.01051} {arXiv:2109.01051 [quant-ph]} \BibitemShut
  {NoStop}%
\bibitem [{\citenamefont {Aaronson}\ and\ \citenamefont
  {Arkhipov}(2013)}]{Aaronson2013a}%
  \BibitemOpen
  \bibfield  {author} {\bibinfo {author} {\bibfnamefont {S.}~\bibnamefont
  {Aaronson}}\ and\ \bibinfo {author} {\bibfnamefont {A.}~\bibnamefont
  {Arkhipov}},\ }\bibfield  {title} {\bibinfo {title} {The {{Computational
  Complexity}} of {{Linear Optics}}},\ }\href
  {https://doi.org/10.4086/toc.2013.v009a004} {\bibfield  {journal} {\bibinfo
  {journal} {Theory Comput.}\ }\textbf {\bibinfo {volume} {9}},\ \bibinfo
  {pages} {143} (\bibinfo {year} {2013})}\BibitemShut {NoStop}%
\bibitem [{\citenamefont {Harrow}\ and\ \citenamefont
  {Montanaro}(2017)}]{Harrow2017}%
  \BibitemOpen
  \bibfield  {author} {\bibinfo {author} {\bibfnamefont {A.~W.}\ \bibnamefont
  {Harrow}}\ and\ \bibinfo {author} {\bibfnamefont {A.}~\bibnamefont
  {Montanaro}},\ }\bibfield  {title} {\bibinfo {title} {Quantum computational
  supremacy},\ }\href {https://doi.org/10.1038/nature23458} {\bibfield
  {journal} {\bibinfo  {journal} {Nature}\ }\textbf {\bibinfo {volume} {549}},\
  \bibinfo {pages} {203} (\bibinfo {year} {2017})}\BibitemShut {NoStop}%
\bibitem [{\citenamefont {Harrow}\ and\ \citenamefont
  {Mehraban}(2018)}]{Harrow2018}%
  \BibitemOpen
  \bibfield  {author} {\bibinfo {author} {\bibfnamefont {A.}~\bibnamefont
  {Harrow}}\ and\ \bibinfo {author} {\bibfnamefont {S.}~\bibnamefont
  {Mehraban}},\ }\bibfield  {title} {\bibinfo {title} {Approximate unitary
  \$t\$-designs by short random quantum circuits using nearest-neighbor and
  long-range gates},\ }\href {http://arxiv.org/abs/1809.06957} {\  (\bibinfo
  {year} {2018})},\ \Eprint {https://arxiv.org/abs/1809.06957}
  {arXiv:1809.06957 [quant-ph]} \BibitemShut {NoStop}%
\bibitem [{\citenamefont {Dalzell}\ \emph {et~al.}(2020)\citenamefont
  {Dalzell}, \citenamefont {{Hunter-Jones}},\ and\ \citenamefont
  {Brand{\~a}o}}]{Dalzell2020}%
  \BibitemOpen
  \bibfield  {author} {\bibinfo {author} {\bibfnamefont {A.~M.}\ \bibnamefont
  {Dalzell}}, \bibinfo {author} {\bibfnamefont {N.}~\bibnamefont
  {{Hunter-Jones}}},\ and\ \bibinfo {author} {\bibfnamefont {F.~G. S.~L.}\
  \bibnamefont {Brand{\~a}o}},\ }\bibfield  {title} {\bibinfo {title} {Random
  quantum circuits anti-concentrate in log depth},\ }\href
  {http://arxiv.org/abs/2011.12277} {\  (\bibinfo {year} {2020})},\ \Eprint
  {https://arxiv.org/abs/2011.12277} {arXiv:2011.12277 [cond-mat,
  physics:quant-ph]} \BibitemShut {NoStop}%
\bibitem [{\citenamefont {Nahum}\ \emph {et~al.}(2017)\citenamefont {Nahum},
  \citenamefont {Ruhman}, \citenamefont {Vijay},\ and\ \citenamefont
  {Haah}}]{Nahum2017}%
  \BibitemOpen
  \bibfield  {author} {\bibinfo {author} {\bibfnamefont {A.}~\bibnamefont
  {Nahum}}, \bibinfo {author} {\bibfnamefont {J.}~\bibnamefont {Ruhman}},
  \bibinfo {author} {\bibfnamefont {S.}~\bibnamefont {Vijay}},\ and\ \bibinfo
  {author} {\bibfnamefont {J.}~\bibnamefont {Haah}},\ }\bibfield  {title}
  {\bibinfo {title} {Quantum {{Entanglement Growth}} under {{Random Unitary
  Dynamics}}},\ }\href {https://doi.org/10.1103/PhysRevX.7.031016} {\bibfield
  {journal} {\bibinfo  {journal} {Phys. Rev. X}\ }\textbf {\bibinfo {volume}
  {7}},\ \bibinfo {pages} {031016} (\bibinfo {year} {2017})}\BibitemShut
  {NoStop}%
\bibitem [{\citenamefont {Nahum}\ \emph {et~al.}(2018)\citenamefont {Nahum},
  \citenamefont {Vijay},\ and\ \citenamefont {Haah}}]{Nahum2018}%
  \BibitemOpen
  \bibfield  {author} {\bibinfo {author} {\bibfnamefont {A.}~\bibnamefont
  {Nahum}}, \bibinfo {author} {\bibfnamefont {S.}~\bibnamefont {Vijay}},\ and\
  \bibinfo {author} {\bibfnamefont {J.}~\bibnamefont {Haah}},\ }\bibfield
  {title} {\bibinfo {title} {Operator {{Spreading}} in {{Random Unitary
  Circuits}}},\ }\href {https://doi.org/10.1103/PhysRevX.8.021014} {\bibfield
  {journal} {\bibinfo  {journal} {Phys. Rev. X}\ }\textbf {\bibinfo {volume}
  {8}},\ \bibinfo {pages} {021014} (\bibinfo {year} {2018})}\BibitemShut
  {NoStop}%
\bibitem [{\citenamefont {{Hunter-Jones}}(2019)}]{Hunter-Jones2019}%
  \BibitemOpen
  \bibfield  {author} {\bibinfo {author} {\bibfnamefont {N.}~\bibnamefont
  {{Hunter-Jones}}},\ }\bibfield  {title} {\bibinfo {title} {Unitary designs
  from statistical mechanics in random quantum circuits},\ }\href
  {http://arxiv.org/abs/1905.12053} {\  (\bibinfo {year} {2019})},\ \Eprint
  {https://arxiv.org/abs/1905.12053} {arXiv:1905.12053 [cond-mat,
  physics:hep-th, physics:quant-ph]} \BibitemShut {NoStop}%
\bibitem [{\citenamefont {Barak}\ \emph {et~al.}(2020)\citenamefont {Barak},
  \citenamefont {Chou},\ and\ \citenamefont {Gao}}]{Barak2020}%
  \BibitemOpen
  \bibfield  {author} {\bibinfo {author} {\bibfnamefont {B.}~\bibnamefont
  {Barak}}, \bibinfo {author} {\bibfnamefont {C.-N.}\ \bibnamefont {Chou}},\
  and\ \bibinfo {author} {\bibfnamefont {X.}~\bibnamefont {Gao}},\ }\bibfield
  {title} {\bibinfo {title} {Spoofing {{Linear Cross-Entropy Benchmarking}} in
  {{Shallow Quantum Circuits}}},\ }\href {http://arxiv.org/abs/2005.02421} {\
  (\bibinfo {year} {2020})},\ \Eprint {https://arxiv.org/abs/2005.02421}
  {arXiv:2005.02421 [quant-ph]} \BibitemShut {NoStop}%
\bibitem [{\citenamefont {Dalzell}\ \emph {et~al.}(2021)\citenamefont
  {Dalzell}, \citenamefont {{Hunter-Jones}},\ and\ \citenamefont
  {Brand{\~a}o}}]{Dalzell2021}%
  \BibitemOpen
  \bibfield  {author} {\bibinfo {author} {\bibfnamefont {A.~M.}\ \bibnamefont
  {Dalzell}}, \bibinfo {author} {\bibfnamefont {N.}~\bibnamefont
  {{Hunter-Jones}}},\ and\ \bibinfo {author} {\bibfnamefont {F.~G. S.~L.}\
  \bibnamefont {Brand{\~a}o}},\ }\bibfield  {title} {\bibinfo {title} {Random
  quantum circuits transform local noise into global white noise},\ }\href
  {http://arxiv.org/abs/2111.14907} {\  (\bibinfo {year} {2021})},\ \Eprint
  {https://arxiv.org/abs/2111.14907} {arXiv:2111.14907 [quant-ph]} \BibitemShut
  {NoStop}%
\bibitem [{\citenamefont {Bremner}\ \emph {et~al.}(2016)\citenamefont
  {Bremner}, \citenamefont {Montanaro},\ and\ \citenamefont
  {Shepherd}}]{Bremner2016}%
  \BibitemOpen
  \bibfield  {author} {\bibinfo {author} {\bibfnamefont {M.~J.}\ \bibnamefont
  {Bremner}}, \bibinfo {author} {\bibfnamefont {A.}~\bibnamefont {Montanaro}},\
  and\ \bibinfo {author} {\bibfnamefont {D.~J.}\ \bibnamefont {Shepherd}},\
  }\bibfield  {title} {\bibinfo {title} {Average-case complexity versus
  approximate simulation of commuting quantum computations},\ }\href
  {https://doi.org/10.1103/PhysRevLett.117.080501} {\bibfield  {journal}
  {\bibinfo  {journal} {Phys. Rev. Lett.}\ }\textbf {\bibinfo {volume} {117}},\
  \bibinfo {pages} {080501} (\bibinfo {year} {2016})}\BibitemShut {NoStop}%
\bibitem [{\citenamefont {Krovi}(2022)}]{Krovi2022}%
  \BibitemOpen
  \bibfield  {author} {\bibinfo {author} {\bibfnamefont {H.}~\bibnamefont
  {Krovi}},\ }\bibfield  {title} {\bibinfo {title} {Average-case hardness of
  estimating probabilities of random quantum circuits with a linear scaling in
  the error exponent}\ }\href {https://doi.org/10.48550/arXiv.2206.05642}
  {10.48550/arXiv.2206.05642} (\bibinfo {year} {2022})\BibitemShut {NoStop}%
\bibitem [{\citenamefont {Kondo}\ \emph {et~al.}(2021)\citenamefont {Kondo},
  \citenamefont {Mori},\ and\ \citenamefont {Movassagh}}]{Kondo2021}%
  \BibitemOpen
  \bibfield  {author} {\bibinfo {author} {\bibfnamefont {Y.}~\bibnamefont
  {Kondo}}, \bibinfo {author} {\bibfnamefont {R.}~\bibnamefont {Mori}},\ and\
  \bibinfo {author} {\bibfnamefont {R.}~\bibnamefont {Movassagh}},\ }\bibfield
  {title} {\bibinfo {title} {Fine-grained analysis and improved robustness of
  quantum supremacy for random circuit sampling},\ }\href
  {http://arxiv.org/abs/2102.01960} {\  (\bibinfo {year} {2021})},\ \Eprint
  {https://arxiv.org/abs/2102.01960} {arXiv:2102.01960 [quant-ph]} \BibitemShut
  {NoStop}%
\bibitem [{\citenamefont {Bouland}\ \emph {et~al.}(2019)\citenamefont
  {Bouland}, \citenamefont {Fefferman}, \citenamefont {Nirkhe},\ and\
  \citenamefont {Vazirani}}]{Bouland2019}%
  \BibitemOpen
  \bibfield  {author} {\bibinfo {author} {\bibfnamefont {A.}~\bibnamefont
  {Bouland}}, \bibinfo {author} {\bibfnamefont {B.}~\bibnamefont {Fefferman}},
  \bibinfo {author} {\bibfnamefont {C.}~\bibnamefont {Nirkhe}},\ and\ \bibinfo
  {author} {\bibfnamefont {U.}~\bibnamefont {Vazirani}},\ }\bibfield  {title}
  {\bibinfo {title} {On the complexity and verification of quantum random
  circuit sampling},\ }\href {https://doi.org/10.1038/s41567-018-0318-2}
  {\bibfield  {journal} {\bibinfo  {journal} {Nat. Phys.}\ }\textbf {\bibinfo
  {volume} {15}},\ \bibinfo {pages} {159} (\bibinfo {year} {2019})}\BibitemShut
  {NoStop}%
\bibitem [{\citenamefont {Movassagh}(2019)}]{Movassagh2019}%
  \BibitemOpen
  \bibfield  {author} {\bibinfo {author} {\bibfnamefont {R.}~\bibnamefont
  {Movassagh}},\ }\bibfield  {title} {\bibinfo {title} {Quantum supremacy and
  random circuits},\ }\href {http://arxiv.org/abs/1909.06210} {\  (\bibinfo
  {year} {2019})},\ \Eprint {https://arxiv.org/abs/1909.06210}
  {arXiv:1909.06210 [cond-mat, physics:hep-th, physics:math-ph,
  physics:quant-ph]} \BibitemShut {NoStop}%
\bibitem [{\citenamefont {Napp}\ \emph {et~al.}(2019)\citenamefont {Napp},
  \citenamefont {La~Placa}, \citenamefont {Dalzell}, \citenamefont {Brandao},\
  and\ \citenamefont {Harrow}}]{Napp2019}%
  \BibitemOpen
  \bibfield  {author} {\bibinfo {author} {\bibfnamefont {J.}~\bibnamefont
  {Napp}}, \bibinfo {author} {\bibfnamefont {R.~L.}\ \bibnamefont {La~Placa}},
  \bibinfo {author} {\bibfnamefont {A.~M.}\ \bibnamefont {Dalzell}}, \bibinfo
  {author} {\bibfnamefont {F.~G. S.~L.}\ \bibnamefont {Brandao}},\ and\
  \bibinfo {author} {\bibfnamefont {A.~W.}\ \bibnamefont {Harrow}},\ }\bibfield
   {title} {\bibinfo {title} {Efficient classical simulation of random shallow
  {{2D}} quantum circuits},\ }\href {http://arxiv.org/abs/2001.00021} {\
  (\bibinfo {year} {2019})},\ \Eprint {https://arxiv.org/abs/2001.00021}
  {arXiv:2001.00021 [cond-mat, physics:quant-ph]} \BibitemShut {NoStop}%
\bibitem [{\citenamefont {Lipton}(1989)}]{Lipton1989}%
  \BibitemOpen
  \bibfield  {author} {\bibinfo {author} {\bibfnamefont {R.~J.}\ \bibnamefont
  {Lipton}},\ }\bibfield  {title} {\bibinfo {title} {New directions in
  testing},\ }in\ \href {https://doi.org/10.1090/dimacs/002/13} {\emph
  {\bibinfo {booktitle} {Distributed Computing and Cryptography, Proceedings of
  a {{DIMACS}} Workshop, Princeton, New Jersey, {{USA}}, October 4-6, 1989}}},\
  \bibinfo {series} {{{DIMACS}} Series in Discrete Mathematics and Theoretical
  Computer Science}, Vol.~\bibinfo {volume} {2},\ \bibinfo {editor} {edited by\
  \bibinfo {editor} {\bibfnamefont {J.}~\bibnamefont {Feigenbaum}}\ and\
  \bibinfo {editor} {\bibfnamefont {M.}~\bibnamefont {Merritt}}}\ (\bibinfo
  {publisher} {{DIMACS/AMS}},\ \bibinfo {year} {1989})\ pp.\ \bibinfo {pages}
  {191--202}\BibitemShut {NoStop}%
\bibitem [{\citenamefont {Aaronson}\ and\ \citenamefont
  {Chen}(2017)}]{Aaronson2017}%
  \BibitemOpen
  \bibfield  {author} {\bibinfo {author} {\bibfnamefont {S.}~\bibnamefont
  {Aaronson}}\ and\ \bibinfo {author} {\bibfnamefont {L.}~\bibnamefont
  {Chen}},\ }\bibfield  {title} {\bibinfo {title} {Complexity-theoretic
  {{Foundations}} of {{Quantum Supremacy Experiments}}},\ }in\ \href
  {https://doi.org/10.4230/LIPIcs.CCC.2017.22} {\emph {\bibinfo {booktitle}
  {Proceedings of the {{32Nd Computational Complexity Conference}}}}},\
  \bibinfo {series and number} {{{CCC}} '17}\ (\bibinfo  {publisher} {{Schloss
  Dagstuhl\textendash Leibniz-Zentrum fuer Informatik}},\ \bibinfo {address}
  {{Germany}},\ \bibinfo {year} {2017})\ pp.\ \bibinfo {pages}
  {22:1--22:67}\BibitemShut {NoStop}%
\bibitem [{\citenamefont {Aaronson}\ and\ \citenamefont
  {Gunn}(2020)}]{Aaronson2020b}%
  \BibitemOpen
  \bibfield  {author} {\bibinfo {author} {\bibfnamefont {S.}~\bibnamefont
  {Aaronson}}\ and\ \bibinfo {author} {\bibfnamefont {S.}~\bibnamefont
  {Gunn}},\ }\bibfield  {title} {\bibinfo {title} {On the classical hardness of
  spoofing linear cross-entropy benchmarking},\ }\href
  {https://doi.org/10.4086/toc.2020.v016a011} {\bibfield  {journal} {\bibinfo
  {journal} {Theory Comput.}\ }\textbf {\bibinfo {volume} {16}},\ \bibinfo
  {pages} {1} (\bibinfo {year} {2020})}\BibitemShut {NoStop}%
\bibitem [{\citenamefont {Choi}\ \emph {et~al.}(2021)\citenamefont {Choi},
  \citenamefont {Shaw}, \citenamefont {Madjarov}, \citenamefont {Xie},
  \citenamefont {Covey}, \citenamefont {Cotler}, \citenamefont {Mark},
  \citenamefont {Huang}, \citenamefont {Kale}, \citenamefont {Pichler},
  \citenamefont {Brand{\~a}o}, \citenamefont {Choi},\ and\ \citenamefont
  {Endres}}]{Choi2021}%
  \BibitemOpen
  \bibfield  {author} {\bibinfo {author} {\bibfnamefont {J.}~\bibnamefont
  {Choi}}, \bibinfo {author} {\bibfnamefont {A.~L.}\ \bibnamefont {Shaw}},
  \bibinfo {author} {\bibfnamefont {I.~S.}\ \bibnamefont {Madjarov}}, \bibinfo
  {author} {\bibfnamefont {X.}~\bibnamefont {Xie}}, \bibinfo {author}
  {\bibfnamefont {J.~P.}\ \bibnamefont {Covey}}, \bibinfo {author}
  {\bibfnamefont {J.~S.}\ \bibnamefont {Cotler}}, \bibinfo {author}
  {\bibfnamefont {D.~K.}\ \bibnamefont {Mark}}, \bibinfo {author}
  {\bibfnamefont {H.-Y.}\ \bibnamefont {Huang}}, \bibinfo {author}
  {\bibfnamefont {A.}~\bibnamefont {Kale}}, \bibinfo {author} {\bibfnamefont
  {H.}~\bibnamefont {Pichler}}, \bibinfo {author} {\bibfnamefont {F.~G. S.~L.}\
  \bibnamefont {Brand{\~a}o}}, \bibinfo {author} {\bibfnamefont
  {S.}~\bibnamefont {Choi}},\ and\ \bibinfo {author} {\bibfnamefont
  {M.}~\bibnamefont {Endres}},\ }\bibfield  {title} {\bibinfo {title} {Emergent
  {{Randomness}} and {{Benchmarking}} from {{Many-Body Quantum Chaos}}},\
  }\href {http://arxiv.org/abs/2103.03535} {\  (\bibinfo {year} {2021})},\
  \Eprint {https://arxiv.org/abs/2103.03535} {arXiv:2103.03535 [cond-mat,
  physics:physics, physics:quant-ph]} \BibitemShut {NoStop}%
\bibitem [{\citenamefont {Gao}\ \emph {et~al.}(2021)\citenamefont {Gao},
  \citenamefont {Kalinowski}, \citenamefont {Chou}, \citenamefont {Lukin},
  \citenamefont {Barak},\ and\ \citenamefont {Choi}}]{Gao2021}%
  \BibitemOpen
  \bibfield  {author} {\bibinfo {author} {\bibfnamefont {X.}~\bibnamefont
  {Gao}}, \bibinfo {author} {\bibfnamefont {M.}~\bibnamefont {Kalinowski}},
  \bibinfo {author} {\bibfnamefont {C.-N.}\ \bibnamefont {Chou}}, \bibinfo
  {author} {\bibfnamefont {M.~D.}\ \bibnamefont {Lukin}}, \bibinfo {author}
  {\bibfnamefont {B.}~\bibnamefont {Barak}},\ and\ \bibinfo {author}
  {\bibfnamefont {S.}~\bibnamefont {Choi}},\ }\bibfield  {title} {\bibinfo
  {title} {Limitations of {{Linear Cross-Entropy}} as a {{Measure}} for
  {{Quantum Advantage}}},\ }\href {http://arxiv.org/abs/2112.01657} {\
  (\bibinfo {year} {2021})},\ \Eprint {https://arxiv.org/abs/2112.01657}
  {arXiv:2112.01657 [cond-mat, physics:quant-ph]} \BibitemShut {NoStop}%
\bibitem [{\citenamefont {Yunger~Halpern}\ and\ \citenamefont
  {Crosson}(2019)}]{YungerHalpern2019}%
  \BibitemOpen
  \bibfield  {author} {\bibinfo {author} {\bibfnamefont {N.}~\bibnamefont
  {Yunger~Halpern}}\ and\ \bibinfo {author} {\bibfnamefont {E.}~\bibnamefont
  {Crosson}},\ }\bibfield  {title} {\bibinfo {title} {Quantum information in
  the {{Posner}} model of quantum cognition},\ }\href
  {https://doi.org/10.1016/j.aop.2018.11.016} {\bibfield  {journal} {\bibinfo
  {journal} {Annals of Physics}\ }\textbf {\bibinfo {volume} {407}},\ \bibinfo
  {pages} {92} (\bibinfo {year} {2019})}\BibitemShut {NoStop}%
\bibitem [{\citenamefont {Chan}\ \emph {et~al.}(2019)\citenamefont {Chan},
  \citenamefont {Nandkishore}, \citenamefont {Pretko},\ and\ \citenamefont
  {Smith}}]{Chan2019}%
  \BibitemOpen
  \bibfield  {author} {\bibinfo {author} {\bibfnamefont {A.}~\bibnamefont
  {Chan}}, \bibinfo {author} {\bibfnamefont {R.~M.}\ \bibnamefont
  {Nandkishore}}, \bibinfo {author} {\bibfnamefont {M.}~\bibnamefont
  {Pretko}},\ and\ \bibinfo {author} {\bibfnamefont {G.}~\bibnamefont
  {Smith}},\ }\bibfield  {title} {\bibinfo {title} {Weak measurements limit
  entanglement to area law (with possible log corrections)},\ }\href
  {https://doi.org/10.1103/PhysRevB.99.224307} {\bibfield  {journal} {\bibinfo
  {journal} {Phys. Rev. B}\ }\textbf {\bibinfo {volume} {99}},\ \bibinfo
  {pages} {224307} (\bibinfo {year} {2019})}\BibitemShut {NoStop}%
\bibitem [{\citenamefont {Li}\ \emph {et~al.}(2018)\citenamefont {Li},
  \citenamefont {Chen},\ and\ \citenamefont {Fisher}}]{Li2018}%
  \BibitemOpen
  \bibfield  {author} {\bibinfo {author} {\bibfnamefont {Y.}~\bibnamefont
  {Li}}, \bibinfo {author} {\bibfnamefont {X.}~\bibnamefont {Chen}},\ and\
  \bibinfo {author} {\bibfnamefont {M.~P.~A.}\ \bibnamefont {Fisher}},\
  }\bibfield  {title} {\bibinfo {title} {Quantum {{Zeno}} effect and the
  many-body entanglement transition},\ }\href
  {https://doi.org/10.1103/PhysRevB.98.205136} {\bibfield  {journal} {\bibinfo
  {journal} {Phys. Rev. B}\ }\textbf {\bibinfo {volume} {98}},\ \bibinfo
  {pages} {205136} (\bibinfo {year} {2018})}\BibitemShut {NoStop}%
\bibitem [{\citenamefont {Skinner}\ \emph {et~al.}(2019)\citenamefont
  {Skinner}, \citenamefont {Ruhman},\ and\ \citenamefont
  {Nahum}}]{Skinner2019}%
  \BibitemOpen
  \bibfield  {author} {\bibinfo {author} {\bibfnamefont {B.}~\bibnamefont
  {Skinner}}, \bibinfo {author} {\bibfnamefont {J.}~\bibnamefont {Ruhman}},\
  and\ \bibinfo {author} {\bibfnamefont {A.}~\bibnamefont {Nahum}},\ }\bibfield
   {title} {\bibinfo {title} {Measurement-{{Induced Phase Transitions}} in the
  {{Dynamics}} of {{Entanglement}}},\ }\href
  {https://doi.org/10.1103/PhysRevX.9.031009} {\bibfield  {journal} {\bibinfo
  {journal} {Phys. Rev. X}\ }\textbf {\bibinfo {volume} {9}},\ \bibinfo {pages}
  {031009} (\bibinfo {year} {2019})}\BibitemShut {NoStop}%
\bibitem [{\citenamefont {Gullans}\ and\ \citenamefont
  {Huse}(2020)}]{Gullans2020a}%
  \BibitemOpen
  \bibfield  {author} {\bibinfo {author} {\bibfnamefont {M.~J.}\ \bibnamefont
  {Gullans}}\ and\ \bibinfo {author} {\bibfnamefont {D.~A.}\ \bibnamefont
  {Huse}},\ }\bibfield  {title} {\bibinfo {title} {Dynamical {{Purification
  Phase Transition Induced}} by {{Quantum Measurements}}},\ }\href
  {https://doi.org/10.1103/PhysRevX.10.041020} {\bibfield  {journal} {\bibinfo
  {journal} {Phys. Rev. X}\ }\textbf {\bibinfo {volume} {10}},\ \bibinfo
  {pages} {041020} (\bibinfo {year} {2020})}\BibitemShut {NoStop}%
\bibitem [{\citenamefont {Vijay}(2020)}]{Vijay2020}%
  \BibitemOpen
  \bibfield  {author} {\bibinfo {author} {\bibfnamefont {S.}~\bibnamefont
  {Vijay}},\ }\bibfield  {title} {\bibinfo {title} {Measurement-{{Driven Phase
  Transition}} within a {{Volume-Law Entangled Phase}}},\ }\href
  {http://arxiv.org/abs/2005.03052} {\  (\bibinfo {year} {2020})},\ \Eprint
  {https://arxiv.org/abs/2005.03052} {arXiv:2005.03052 [cond-mat,
  physics:hep-th, physics:quant-ph]} \BibitemShut {NoStop}%
\bibitem [{\citenamefont {Bao}\ \emph {et~al.}(2022)\citenamefont {Bao},
  \citenamefont {Block},\ and\ \citenamefont {Altman}}]{Bao2022}%
  \BibitemOpen
  \bibfield  {author} {\bibinfo {author} {\bibfnamefont {Y.}~\bibnamefont
  {Bao}}, \bibinfo {author} {\bibfnamefont {M.}~\bibnamefont {Block}},\ and\
  \bibinfo {author} {\bibfnamefont {E.}~\bibnamefont {Altman}},\ }\bibfield
  {title} {\bibinfo {title} {Finite time teleportation phase transition in
  random quantum circuits},\ }\href {http://arxiv.org/abs/2110.06963} {\
  (\bibinfo {year} {2022})},\ \Eprint {https://arxiv.org/abs/2110.06963}
  {arXiv:2110.06963 [cond-mat, physics:quant-ph]} \BibitemShut {NoStop}%
\bibitem [{\citenamefont {{M{\"u}ller-Hermes}}\ \emph
  {et~al.}(2016)\citenamefont {{M{\"u}ller-Hermes}}, \citenamefont
  {Stilck~Fran{\c c}a},\ and\ \citenamefont {Wolf}}]{Muller-Hermes2016}%
  \BibitemOpen
  \bibfield  {author} {\bibinfo {author} {\bibfnamefont {A.}~\bibnamefont
  {{M{\"u}ller-Hermes}}}, \bibinfo {author} {\bibfnamefont {D.}~\bibnamefont
  {Stilck~Fran{\c c}a}},\ and\ \bibinfo {author} {\bibfnamefont {M.~M.}\
  \bibnamefont {Wolf}},\ }\bibfield  {title} {\bibinfo {title} {Relative
  entropy convergence for depolarizing channels},\ }\href
  {https://doi.org/10.1063/1.4939560} {\bibfield  {journal} {\bibinfo
  {journal} {J. Math. Phys.}\ }\textbf {\bibinfo {volume} {57}},\ \bibinfo
  {pages} {022202} (\bibinfo {year} {2016})}\BibitemShut {NoStop}%
\bibitem [{\citenamefont {Hirche}\ \emph {et~al.}(2020)\citenamefont {Hirche},
  \citenamefont {Rouz{\'e}},\ and\ \citenamefont {Fran{\c c}a}}]{Hirche2020}%
  \BibitemOpen
  \bibfield  {author} {\bibinfo {author} {\bibfnamefont {C.}~\bibnamefont
  {Hirche}}, \bibinfo {author} {\bibfnamefont {C.}~\bibnamefont {Rouz{\'e}}},\
  and\ \bibinfo {author} {\bibfnamefont {D.~S.}\ \bibnamefont {Fran{\c c}a}},\
  }\bibfield  {title} {\bibinfo {title} {On contraction coefficients, partial
  orders and approximation of capacities for quantum channels},\ }\href
  {https://arxiv.org/abs/2011.05949} {\  (\bibinfo {year} {2020})},\ \Eprint
  {https://arxiv.org/abs/2011.05949} {arXiv:2011.05949} \BibitemShut {NoStop}%
\bibitem [{\citenamefont {{Ben-Or}}\ \emph {et~al.}(2013)\citenamefont
  {{Ben-Or}}, \citenamefont {Gottesman},\ and\ \citenamefont
  {Hassidim}}]{Ben-Or2013}%
  \BibitemOpen
  \bibfield  {author} {\bibinfo {author} {\bibfnamefont {M.}~\bibnamefont
  {{Ben-Or}}}, \bibinfo {author} {\bibfnamefont {D.}~\bibnamefont
  {Gottesman}},\ and\ \bibinfo {author} {\bibfnamefont {A.}~\bibnamefont
  {Hassidim}},\ }\bibfield  {title} {\bibinfo {title} {Quantum
  {{Refrigerator}}},\ }\href {http://arxiv.org/abs/1301.1995} {\  (\bibinfo
  {year} {2013})},\ \Eprint {https://arxiv.org/abs/1301.1995} {arXiv:1301.1995
  [quant-ph]} \BibitemShut {NoStop}%
\bibitem [{\citenamefont {Hangleiter}\ \emph {et~al.}(2018)\citenamefont
  {Hangleiter}, \citenamefont {{Bermejo-Vega}}, \citenamefont {Schwarz},\ and\
  \citenamefont {Eisert}}]{Hangleiter2018}%
  \BibitemOpen
  \bibfield  {author} {\bibinfo {author} {\bibfnamefont {D.}~\bibnamefont
  {Hangleiter}}, \bibinfo {author} {\bibfnamefont {J.}~\bibnamefont
  {{Bermejo-Vega}}}, \bibinfo {author} {\bibfnamefont {M.}~\bibnamefont
  {Schwarz}},\ and\ \bibinfo {author} {\bibfnamefont {J.}~\bibnamefont
  {Eisert}},\ }\bibfield  {title} {\bibinfo {title} {Anticoncentration theorems
  for schemes showing a quantum speedup},\ }\href
  {https://doi.org/10.22331/q-2018-05-22-65} {\bibfield  {journal} {\bibinfo
  {journal} {Quantum}\ }\textbf {\bibinfo {volume} {2}},\ \bibinfo {pages} {65}
  (\bibinfo {year} {2018})}\BibitemShut {NoStop}%
\bibitem [{\citenamefont {Harrow}\ and\ \citenamefont
  {Low}(2009)}]{Harrow2009a}%
  \BibitemOpen
  \bibfield  {author} {\bibinfo {author} {\bibfnamefont {A.~W.}\ \bibnamefont
  {Harrow}}\ and\ \bibinfo {author} {\bibfnamefont {R.~A.}\ \bibnamefont
  {Low}},\ }\bibfield  {title} {\bibinfo {title} {Random {{Quantum Circuits}}
  are {{Approximate}} 2-designs},\ }\href
  {https://doi.org/10.1007/s00220-009-0873-6} {\bibfield  {journal} {\bibinfo
  {journal} {Commun. Math. Phys.}\ }\textbf {\bibinfo {volume} {291}},\
  \bibinfo {pages} {257} (\bibinfo {year} {2009})}\BibitemShut {NoStop}%
\bibitem [{\citenamefont {Brandao}\ \emph {et~al.}(2016)\citenamefont
  {Brandao}, \citenamefont {Harrow},\ and\ \citenamefont
  {Horodecki}}]{Brandao2016}%
  \BibitemOpen
  \bibfield  {author} {\bibinfo {author} {\bibfnamefont {F.~G. S.~L.}\
  \bibnamefont {Brandao}}, \bibinfo {author} {\bibfnamefont {A.~W.}\
  \bibnamefont {Harrow}},\ and\ \bibinfo {author} {\bibfnamefont
  {M.}~\bibnamefont {Horodecki}},\ }\bibfield  {title} {\bibinfo {title} {Local
  random quantum circuits are approximate polynomial-designs},\ }\href
  {https://doi.org/10.1007/s00220-016-2706-8} {\bibfield  {journal} {\bibinfo
  {journal} {Commun. Math. Phys.}\ }\textbf {\bibinfo {volume} {346}},\
  \bibinfo {pages} {397} (\bibinfo {year} {2016})}\BibitemShut {NoStop}%
\bibitem [{\citenamefont {Morimae}(2017)}]{Morimae2017b}%
  \BibitemOpen
  \bibfield  {author} {\bibinfo {author} {\bibfnamefont {T.}~\bibnamefont
  {Morimae}},\ }\bibfield  {title} {\bibinfo {title} {Hardness of classically
  sampling the one-clean-qubit model with constant total variation distance
  error},\ }\href {https://doi.org/10.1103/PhysRevA.96.040302} {\bibfield
  {journal} {\bibinfo  {journal} {Phys. Rev. A}\ }\textbf {\bibinfo {volume}
  {96}},\ \bibinfo {pages} {040302(R)} (\bibinfo {year} {2017})}\BibitemShut
  {NoStop}%
\bibitem [{\citenamefont {Wallman}\ and\ \citenamefont
  {Emerson}(2016)}]{Wallman2016}%
  \BibitemOpen
  \bibfield  {author} {\bibinfo {author} {\bibfnamefont {J.~J.}\ \bibnamefont
  {Wallman}}\ and\ \bibinfo {author} {\bibfnamefont {J.}~\bibnamefont
  {Emerson}},\ }\bibfield  {title} {\bibinfo {title} {Noise tailoring for
  scalable quantum computation via randomized compiling},\ }\href
  {https://doi.org/10.1103/PhysRevA.94.052325} {\bibfield  {journal} {\bibinfo
  {journal} {Phys. Rev. A}\ }\textbf {\bibinfo {volume} {94}},\ \bibinfo
  {pages} {052325} (\bibinfo {year} {2016})}\BibitemShut {NoStop}%
\end{thebibliography}
\end{document}